\documentclass[runningheads]{llncs}
\usepackage[T1]{fontenc}
\usepackage{graphicx}
\usepackage[colorlinks=true, linkcolor=blue, citecolor=blue, urlcolor=blue]{hyperref}
\usepackage{xcolor}

\usepackage{booktabs}
\usepackage{pifont}
\usepackage{tabularx,float}
\usepackage{multirow}
\usepackage{array}
\usepackage{enumitem}
\usepackage{amsmath,amsfonts,xspace,url}
\usepackage{pgfgantt}
\usepackage{ragged2e}
\usepackage{pdflscape}
\usepackage{algorithm}
\usepackage{algpseudocode}
\usepackage{dashbox}
\usepackage{float}
\usepackage{tikz}
\usetikzlibrary{positioning,arrows.meta}

\makeatletter
\renewcommand{\ALG@beginalgorithmic}{\footnotesize}
\makeatother

\newcolumntype{R}{>{\raggedleft\arraybackslash}p{0.8em}}

\DeclareMathSymbol{\Zz}{\mathalpha}{AMSb}{"5A}
\newcommand{\Z}[1]{\ensuremath{\Zz_{#1}} }
\newcommand{\Zs}[1]{\ensuremath{\Zz^{\ast}_{#1}}}
\newcommand{\G}{\ensuremath{\mathbb{G}}\xspace}
\newcommand{\F}[1]{\ensuremath{\mathcal{F}_{#1}}}

\newcommand{\C}{\ensuremath{\mathcal{C}}\xspace}

\begin{document}
\title{From $\textsf{TS-SUF-2}$ to $\textsf{TS-SUF-4}$: Practical Security Enhancements for $\textsf{FROST2}$ Threshold Signatures}
%
%
\author{
Will Wang \and
Syh-Yuan Tan \and
Ryan Chow \and
Chanson Chan \and 
Martin Zhao}
\authorrunning{Will et al.}
%
\institute{Solv Protocol, Singapore\\
\email{\{will,ryan,chanson,martin\}@solv.finance, sytansy@gmail.com}}
\maketitle              
\begin{abstract}
  Threshold signature schemes play a vital role in securing digital assets within blockchain and distributed systems. \textsf{FROST2} stands out as a practical threshold Schnorr signature scheme, noted for its efficiency and compatibility with standard verification processes. However, under the one-more discrete logarithm assumption, with static corruption and centralized key generation settings, \textsf{FROST2} has been shown by Bellare et al. (in CRYPTO 2022) to achieve only \textsf{TS-SUF-2} security, which is a consequence of its vulnerability to \textsf{TS-UF-3} attacks.

In this paper, we address this security limitation by presenting an enhanced variant of \textsf{FROST2}, namely, \textsf{FROST2}\texttt{+} which achieves the \textsf{TS-SUF-4} security level under the same computational assumptions as the original \textsf{FROST2}. 
\textsf{FROST2}\texttt{+} strengthens \textsf{FROST2} by integrating additional pre-processing token verifications that help mitigate \textsf{TS-UF-3} and \textsf{TS-UF-4} vulnerabilities while maintaining practical efficiency. 
We show that \textsf{FROST2}\texttt{+} can achieve \textsf{TS-SUF-4} security not only under the same conditions as the original \textsf{FROST2} analysis, but also when initialized with a distributed key generation protocol such as \textsf{PedPoP}. 
Our benchmark using ZCash's $\textsf{FROST}$ library shows that the performance of \textsf{FROST2}\texttt{+} is comparable to \textsf{FROST2} and about $64-79\%$ faster than \textsf{FROST} when precomputation is enabled.

\end{abstract}

\section{Introduction}
\label{sec:intro}
The Schnorr signature scheme \cite{Sch91} has emerged as a cornerstone for blockchain applications due to its efficient signing and verification processes, compact signatures, and support for aggregation. In this context, Schnorr-based threshold signatures are particularly critical for threshold signing in blockchain systems \cite{sbtc,Chainflip22,arch2024,safe2024}, with \textsf{FROST} \cite{FROST} establishing itself as a benchmark and currently undergoing IETF standardization \cite{rfc9591}. Furthermore, NIST's recent call for contributions on threshold multi-party schemes \cite{nist23} has intensified research interest in \textsf{FROST}.

Research related to \textsf{FROST} can be categorized into two main areas: the development of new schemes that offer enhanced efficiency \cite{CKM21,BTZ22,BCKMTZ22,KG25}, and improvements in security \cite{ROAST,CGRS23,BCLTZ25,Chen25,BDLR25,NRT25,BGCRSY25}. The latter category can be further divided into specific areas of security analysis focused on distributed key generation \cite{CGRS23,CKM21,BCKMTZ22} and adaptive security \cite{CKM23,CKKTZ25,Chen25,CS25,BCLTZ25,NRT25,CKM25,BGCRSY25}. Several works \cite{BTZ22,BCKMTZ22,CPZ24,CKKTZ25,CS25} delve deeper into the security notions relevant to threshold signature schemes.

The \textsf{TS-(S)UF-}$i$ security framework \cite{BTZ22,BCKMTZ22} represents a recent effort to rigorously formalize the security of threshold signature schemes. This framework groups a message $M$ along with its auxiliary signing information as a leader request $(lr)$ and classifies unforgeability levels according to adversarial capabilities, ranging from the weakest (\textsf{TS-UF-0}, which resists less capable attackers) to the strongest (\textsf{TS-UF-4}, which resists more capable attackers). While the framework typically describes forgeries in terms of triviality, let $c<t$  denotes the number of corrupted signers, we outline the attackers’ abilities at each level for clarity:

\begin{enumerate}
    \item \textsf{TS-UF-0}: The threshold scheme can withstand an attacker who cannot obtain any partial signatures from honest signers for the target message $M$.

    \item \textsf{TS-UF-1}: The threshold scheme can withstand an attacker at \textsf{TS-UF-0} who also queries at most $t - c - 1$ partial signatures on $M$ from honest signers.

    \item \textsf{TS-UF-2}: The threshold scheme can withstand an attacker at \textsf{TS-UF-1} who also queries leader requests $(lr)$ on $M$ and obtain at most $t - c - 1$ honest partial signatures per $lr$.

    \item \textsf{TS-UF-3}: The threshold scheme can resist an attacker at \textsf{TS-UF-2} who also corrupts pre-processing tokens $(pp_i)$ for honest signers $\mathsf{id}_i \in lr.\mathcal{S}$ in the signing set, but not for all non-responding honest signers in the same set.

    \item \textsf{TS-UF-4}: The threshold scheme can resist an attacker at \textsf{TS-UF-3} who corrupts pre-processing tokens $(pp_i)$ for all honest signers $\mathsf{id}_i \in lr.\mathcal{S}$ in the signing set, though it cannot obtain partial signatures from every honest signer in $lr.\mathcal{S}$.
\end{enumerate}
where the descriptions of \textsf{TS-UF-3} and \textsf{TS-UF-4} refer to the adversary capabilities in echo schemes \cite{BCKMTZ22,BTZ22,CPZ24} like \textsf{FROST}. 
In such schemes, an honest signer $\mathsf{id}_i$ only produces a partial signature if the leader request $(lr)$ contains $\mathsf{id}_i$’s honestly generated pre-processing token $pp_i \in lr.\mathsf{PP}(\mathsf{id}_i)$. 
The framework also introduces the concept of strong unforgeability (\textsf{TS-SUF-}$i$), which ties a request $lr$ to each unique signature, beyond just the message $M$.

Moreover, the framework provides valuable insights into the relationships among security levels. For instance, it demonstrates that achieving \textsf{TS-UF-0} security implies \textsf{TS-UF-1} security, risking a security loss of $\binom{n}{t-1}$ for the total number of signers $n$ and threshold $t$.
Additionally, it shows that combining a threshold scheme achieving \textsf{TS-UF-3} security with a standard digital signature scheme with strong existential unforgeability under chosen-message attacks (\textsf{SUF-CMA}) results in a threshold scheme with \textsf{TS-UF-4} security.

Within this security framework, both \textsf{FROST} and \textsf{FROST2} achieve \textsf{TS-UF-0} (and thereby \textsf{TS-UF-1}) security in the distributed key generation (DKG) setting \cite{FROST,CKM21,BCKMTZ22}, which represents the baseline security level within the \textsf{TS-UF-}$i$ hierarchy.
Under the centralized key generation (CKG) setting, \textsf{FROST} achieves \textsf{TS-SUF-3} security and can be further improved to attain the highest \textsf{TS-SUF-4} security level by incorporating an additional digital signature scheme for each signer \cite{BTZ22,BCKMTZ22,CPZ24}. However, this approach introduces additional computational overhead and implementation complexity. 
More critically for our work, \textsf{FROST2} \cite{CKM21,BTZ22,BCKMTZ22} has been shown to achieve only \textsf{TS-SUF-2} security and does not qualified for this security upgrade.

The limitation of \textsf{FROST2} to \textsf{TS-SUF-2} security arises from a fundamental vulnerability to \textsf{TS-UF-3} attacks, as demonstrated by Bellare et al. \cite{BTZ22,BCKMTZ22}. This attack is reminiscent of the rogue public key attack in multi-signature schemes, exploiting a weakness in \textsf{FROST2}'s pre-processing token mechanism to forge a Schnorr signature. Consequently, this enables a forgery under an unauthorized signer set without needing to know the partial signatures of all honest signers, thereby compromising \textsf{TS-UF-3} security (see Appendix \ref{apdx:ts-uf-3_attack}).

Enhancements to \textsf{FROST2} have been proposed, but none have achieved \textsf{TS-SUF-3} security or higher. For instance, \textsf{ROAST} \cite{ROAST} introduces token aggregation and signature share verification, improving robustness while acknowledging security reduction to the original \textsf{FROST2}. 
Its successor, \textsf{Olaf} \cite{CGRS23}, combines \textsf{ROAST} with a simplified DKG protocol (\textsf{SimplPedPoP}) from \textsf{FROST}'s \textsf{PedPoP}, proving its \textsf{TS-SUF-2} security. 
Recent theoretical results \cite{CKKTZ25,CKM25} indicate that \textsf{FROST}, \textsf{FROST2}, and \textsf{FROST3} (or \textsf{ROAST}) are either unforgeable with full adaptive security in the algebraic group model (AGM), under the assumptions of Algebraic One-More Discrete Logarithm (AOMDL) and Low-Dimensional Vector Representation (LDVR); or adaptively secure with a maximum of $t/2$ corruptions, assuming AOMDL in the random oracle model (ROM).

\textsf{FROST-Mask} \cite{Chen25}, which can be viewed as an extension of \textsf{FROST2} with masking mechanisms applied to the pre-processing tokens and partial signatures, has been proven unforgeable with adaptive security under the ROM. Conversely, \textsf{ms-FROST} \cite{BCLTZ25} and \textsf{FaFROST} \cite{BGCRSY25} apply masking solely to the partial signature, proven to achieve \textsf{TS-UF-4} and \textsf{adp-TS-UF}, respectively, with adaptive security under the weaker AGM. 

The security notions \textsf{TS-SUF-3} and \textsf{TS-SUF-4} are important in practice. The \textsf{TS-SUF-3} security protects against premature leakage of a pre-processing token, and the \textsf{TS-SUF-4} security prevents an honest signer from being wrongfully framed for a forgery. In this work, we investigate whether \textsf{FROST2}, which improves efficiency over \textsf{FROST} and reduces communication rounds compared to several later three-round variants \cite{CKM21,BDLR25,NRT25}, can achieve the same \textsf{TS-SUF-3} and \textsf{TS-SUF-4} security levels as \textsf{FROST} in the ROM.

\subsection{Our Contributions}
We proposed an improved variant of \textsf{FROST2} that achieves \textsf{TS-SUF-4} security, the highest tier in Bellare et al.'s hierarchical framework, while maintaining the protocol’s efficiency and its two-round structure.

The proposed \textsf{FROST2}\texttt{+} addresses the \textsf{TS-UF-3} vulnerability of \textsf{FROST2} by introducing a lightweight token authentication mechanism for pre-processing tokens.
In contrast to Bellare et al.'s \textsf{ATS[TS, DS]} transformation \cite{BTZ22,BCKMTZ22},
which requires each signer to generate and manage an additional digital signature key-pair,
our approach authenticates the tokens using only the signing key share and the secret exponents already presented in the pre-processing tokens $(R,S)$.
We establish the \textsf{TS-SUF-4} security of \textsf{FROST2}\texttt{+} under the one-more discrete logarithm (OMDL) assumption within the ROM, in the static corruption and centralized key generation (CKG) settings.

Subsequently, we show that our proposed schemes maintain their security levels when transitioning from the CKG setting to a distributed key generation (DKG) setting, specifically using \textsf{PedPoP} DKG as an instantiation. 
We also explain how to speed up the partial signing algorithm by pre-verifying the pre-processing tokens, which is not bind to the signing message.
Moreover, similar to \textsf{FROST}'s offline pre-processing, our token verification mechanism facilitates per-signer batching across multiple signing sessions. 
This design choice preserves the efficiency of \textsf{FROST}-style pre-processing, making our schemes particularly advantageous for high-throughput, leaderless environments where offline token preparation is critical.

Finally, we show that the proposed \textsf{FROST2}\texttt{+}, while offering stronger security, achieves better efficiency than \textsf{FROST} and matches the performance of \textsf{FROST2}, as detailed in our performance evaluation (Section~\ref{sec:performance}).

\paragraph{Organization.} This paper is structured as follows. 
In Section \ref{sec:prelim}, we present the necessary mathematical tools. 
We then outline the proposed \textsf{FROST2}\texttt{+} scheme in Section \ref{sec:our_construction}, followed by the security analysis in Section \ref{sec:security_analysis}. 
Next, we 
present the performance evaluation in Section \ref{sec:performance} and conclude the paper in Section \ref{sec:conclude}.

\section{Preliminaries}
\label{sec:prelim}
We briefly recall the required mathematical tools in this section.

\subsection{One-More Discrete Logarithm (OMDL) Assumption}
\begin{definition}
Let $\G$ be a cyclic group of prime order $p$ with generator $g$, and let \emph{\textsc{Chall}} and \emph{\textsc{Dlog}} be oracles such that:
\begin{itemize}
\item \emph{\textsc{Chall}}$()$ returns a random challenge $A = g^a$ for $a \xleftarrow{\$} \Zs{p}$
\item \emph{\textsc{Dlog}}$(B)$ returns $b\in\Zs{p}$ such that $B = g^b$ for any $B \in \G$
\end{itemize}
The \emph{OMDL} assumption states that for any probabilistic polynomial-time adversary \C making at most $q-1$ queries to \emph{\textsc{Dlog}}, the advantage
\[
\mathsf{Adv}_{\G}^{\mathsf{omdl}}(\C) = \Pr\left[\begin{array}{l}
(A_1,\ldots,A_q) \leftarrow \emph{\textsc{Chall}}(), \\
(a_1,\ldots,a_q) \leftarrow \C^{\emph{\textsc{Chall},\textsc{Dlog}}}(g) \\
\forall i: g^{a_i} = A_i \land \C \text{ never queried } \emph{\textsc{Dlog}}(A_i)
\end{array}\right]
\]
is negligible.
\end{definition}

\subsection{Lagrange Coefficient} 
Let $f(\cdot)$ be a random $(t-1)$-degree polynomial such that $f(x_i)=s_i$.
The polynomial $f(\cdot)$ can be reconstructed by having the knowledge of $t$-pair values of $(x_i,s_i)$:
\begin{align*}
f(\cdot)=\sum_{x_i\in\mathcal{S}}s_{i}\Delta_{x_i,\mathcal{S}}(\cdot)
\end{align*}
where $\mathcal{S}=\{x_1,x_2,\ldots,x_{t}\}$.
\begin{definition}
The Lagrange coefficient $\Delta_{x_i,\mathcal{S}}$ is defined as:
\begin{align*}
\Delta_{x_i,\mathcal{S}}(Z)=\prod_{x_j\in S,j\neq i}\frac{Z-x_j}{x_i-x_j}
\end{align*}
where $x_i$ and $x_j$ are elements in the set $\mathcal{S}$.
\end{definition}

\subsection{Threshold Signature Scheme}
\label{sec:ts_definition}
To synchronize with Bellare et al.'s \textsf{TS-(S)UF-}$i$ security framework \cite{BTZ22,BCKMTZ22}, we adopt their scheme definition while making some adjustments, which includes the addition of a \textsf{Setup} algorithm and modifications to certain inputs for each algorithm.
We define a non-interactive threshold signature scheme as \textsf{TS=\{Setup, Kg, SPP, LPP, LR, PS, Agg, Vf, SVf\}}, characterized as follows:\\

\noindent
$\mathsf{Setup}(1^{\lambda})\rightarrow \mathsf{par}$: The setup algorithm takes in a security parameter $1^{\lambda}$ and produces public parameters \textsf{par}, which are implicitly provided to all other algorithms.\\

\noindent
$\mathsf{Kg}(\textsf{par},t,\mathsf{ns})\rightarrow \{\mathsf{st}_i=(\tilde{X},\mathsf{aux},\mathsf{sk})\}_{i=1}^{\mathsf{ns}}$: The key generation algorithm takes in threshold $t$ and number of signers $\mathsf{ns}$ such that $t\leq \mathsf{ns}$, initializes the state $\mathsf{st}_i$ for every signer $\mathsf{id}_i\in\{\mathsf{id}_1,\ldots,\mathsf{id}_{\mathsf{ns}}\}$ with joint public key $\tilde{X}$, public auxiliary information $\mathsf{aux}$, and a signing key share $\mathsf{sk}$. \\

\noindent
$\mathsf{SPP}(\mathsf{st}_i)\rightarrow(pp, \mathsf{st}_i)$: The signing pre-processing algorithm takes as input the state $\mathsf{st}_i$ of signer $\mathsf{id}_i$. 
It outputs a pre-processing token $pp$ and updates the signer’s state.\\

\noindent
$\mathsf{LPP}(\mathsf{id}_i, pp, \mathsf{st}_0) \rightarrow \mathsf{st}_0$: The leader pre-processing algorithm takes the pre-processing token $pp$ from signer $\mathsf{id}_i$ and the leader’s state $\mathsf{st}_0$, updating $\mathsf{st}_0$ to incorporate the token.\\

\noindent
$\mathsf{LR}(M, \mathcal{S}, \mathsf{st}_0)\rightarrow(lr, \mathsf{st}_0)$: 
This leader signing-request algorithm generates a leader request $lr$ based on a message $M$, a signing set $\mathcal{S}$ with at least $t$ signers, and the leader’s state $\mathsf{st}_0$.\\

\noindent
$\mathsf{PS}(lr, \mathsf{st}_i)\rightarrow (psig_i, \mathsf{st}_i)$: The signer partial-signing algorithm takes the leader request $lr$, the state $\mathsf{st}_i$ of a signer $\mathsf{id}_i$, generating a partial signature $psig_i$ and updating state $\mathsf{st}_i$.\\

\noindent
$\mathsf{Agg}(lr, \{psig_i\}_{\mathsf{id}_i\in lr.\mathcal{S}}, \mathsf{st}_0)\rightarrow (\sigma, \mathsf{st}_0)$: The leader partial-signature aggregation algorithm aggregates the partial signatures received from the signing set $\mathcal{S}$ into a final signature $\sigma$ for the message $M$ while updating the leader's state $\mathsf{st}_0$.\\

\noindent
$\mathsf{Vf}(\tilde{X}, M, \sigma) \rightarrow 0\text{ or }1$: The verification algorithm checks the validity of the signature $\sigma$ for the message $M$ using the public key $\tilde{X}$, returning 1 if valid and 0 otherwise.\\

\noindent
$\mathsf{SVf}(\tilde{X}, lr, \sigma) \rightarrow 0\text{ or }1$: The strong verification algorithm checks the validity of the signature $\sigma$ for the leader request $lr$ using the public key $\tilde{X}$, returning 1 if valid and 0 otherwise.

\subsection{Security Model}
\label{sec:security_model}

\begin{table}[h]
\centering
\caption{Trivial forgery predicates for \textsf{TS‐(S)UF‐}$i$ \cite{BTZ22,BCKMTZ22} with Cremers et al.'s recommendation \cite{CPZ24}.}
\label{tab:trivial-forgery}
\begin{tabular}{l l}
\toprule
\textbf{Predicate} & \textbf{Definition} \\
\midrule
$\textsf{tf}_0(M)$         & $S_1(M) \neq \emptyset$ \\
$\textsf{tf}_1(M)$         & $|S_1(M)| \geq t - |CS|$ \\
$\textsf{tf}_2(lr)$       & $|S_2(lr)| \geq t - |CS|$ \\
$\textsf{tf}_3(lr)$       & $\textsf{tf}_2(lr) \wedge S_2(lr) = S_3(lr)$ \\
$\textsf{tf}_4(lr)$       & $\textsf{tf}_2(lr) \wedge S_2(lr) = S_3(lr) = S_4(lr)$ \\ \addlinespace
$\textsf{tsf}_2(lr, \tilde{X}, \sigma)$ & $\textsf{tf}_2(lr) \wedge \textsf{SVf}(\tilde{X}, lr, \sigma)$ \\
$\textsf{tsf}_3(lr, \tilde{X}, \sigma)$ & $\textsf{tf}_3(lr) \wedge \textsf{SVf}(\tilde{X}, lr, \sigma)$ \\
$\textsf{tsf}_4(lr, \tilde{X}, \sigma)$ & $\textsf{tf}_4(lr) \wedge \textsf{SVf}(\tilde{X}, lr, \sigma)$ \\
\bottomrule
\end{tabular}
\end{table}

\begin{figure}[h!]
  \centering
  \begin{tabular}{|l|}
    \hline
    \begin{minipage}{\linewidth}\small
\underline{Games $\mathbf{G}^{\text{ts-uf-}i}_{\mathsf{TS}}$ ($i=0,1,2,3,4$) and $\mathbf{G}^{\text{ts-suf-}i}_{\mathsf{TS}}$ ($i=2,3,4$)}\\

\begin{algorithmic}[1]
\Statex \hspace*{-\algorithmicindent}\textsc{Init}$(CS)$:   
\State \textsf{Require}: $CS \subseteq\{\mathsf{id}_1,\ldots, \mathsf{id}_\mathsf{ns}\}$ and $|CS|<t$ \Comment{Set of corrupted signers}
\State $\textsf{h}\xleftarrow{\$}\textsf{TS.HF};(\tilde{X},aux,\mathsf{sk}_{1},\dots,\mathsf{sk}_{\mathsf{ns}})\xleftarrow{\$}\textsf{Kg}$ \Comment{Set of hash functions \textsf{TS.HF}}
\State $HS\leftarrow\{\mathsf{id}_1,\ldots,\mathsf{id}_\mathsf{ns}\}\setminus CS$ \Comment{Set of honest signers}
\For{$\mathsf{id}_i\in HS$}
\State $\textsf{st}_i.\textsf{sk}\leftarrow x_i;\;\textsf{st}_i.\tilde{X}\leftarrow \tilde{X};\;\textsf{st}_i.\textsf{aux}\leftarrow aux$
\EndFor
\State \Return $\tilde{X},aux,\{x_i\}_{\mathsf{id}_i\in CS}$

\hfill
\Statex \hspace*{-\algorithmicindent}\textsc{PPO}($\mathsf{id}_i$):
\State \textsf{Require}: $\mathsf{id}_i\in HS$
\State $(pp,\mathsf{st}_i)\xleftarrow{\$}\textsf{SPP}(\mathsf{st}_i);\;PPL\leftarrow PPL\cup\{(\mathsf{id}_i,pp)\}$
\State \Return $pp$

\hfill
\Statex \hspace*{-\algorithmicindent}\textsc{PSignO}($\mathsf{id}_i,lr$):
\State \textsf{Require}: $lr.\mathcal{S}\subseteq\{\mathsf{id}_1,\ldots,\mathsf{id}_\mathsf{ns}\}$, $M=lr.\mathsf{msg}\in\{0,1\}^*$ and $\mathsf{id}_i\in HS$
\State $(psig,\textsf{st}_i)\xleftarrow{\$}\textsf{PS}(lr,\textsf{st}_i);L\leftarrow L\cup\{lr\}$
\If{$psig\neq\bot$}
\State $S_1(M)\leftarrow S_1(M)\cup\{\mathsf{id}_i\};S_2(lr)\leftarrow S_2(lr)\cup\{\mathsf{id}_i\}$
\EndIf
\State \Return $psig$

\hfill
\Statex \hspace*{-\algorithmicindent}\textsc{RO}$[\mathsf{h}](x):$
\State \Return $\mathsf{h}(x)$

\hfill
\Statex \hspace*{-\algorithmicindent}\textsc{Fin}($M,\sigma$):
\For{all $lr\in L$}
\State $S_3(lr)\leftarrow\{\mathsf{id}_i\in HS\cap lr.\mathcal{S}:lr\textsf{.PP}(\mathsf{id}_i)\in PPL\};S_4(lr)\leftarrow HS\cap lr.\mathcal{S}$
\EndFor
\If{$\neg \textsf{Vf}(\tilde{X},M,\sigma)$}
\State \Return \textsf{false};
\EndIf
\If{Game $\mathbf{G}^{\text{ts-uf-}i}_{\mathsf{TS}}$ for $i=0,1$}
\State \Return $\neg \textsf{tf}_i(M)$ 
\EndIf
\If{Game $\mathbf{G}^{\text{ts-uf-}i}_{\mathsf{TS}}$ for $i=2,3,4$}
\State \Return $\neg\exists lr(lr\textsf{.msg}=M\wedge \textsf{tf}_i(lr))$ 
\EndIf
\If{Game $\mathbf{G}^{\text{ts-suf-}i}_{\mathsf{TS}}$}
\State \Return $\neg\exists lr (lr\textsf{.msg}=M\wedge \textsf{tsf}_i(lr,\tilde{X},\sigma))$ 
\EndIf
\end{algorithmic}
    \end{minipage}\\
    \hline
  \end{tabular}
  \caption{Security games defining \textsf{TS‐(S)UF‐}$i$ \cite{BCKMTZ22,BTZ22,CPZ24}.}
  \label{fig:ts-uf-suf}
\end{figure}

To address a shortcoming in the original hierarchy of Bellare et al.'s hierarchy framework \cite{BCKMTZ22}, where \textsf{TS-(S)UF-}4 does not imply \textsf{TS-(S)UF-}3, and vice versa, due to incomparable trivial forgery conditions, we adopt the strengthened trivial forgery condition in Table \ref{tab:trivial-forgery} for $\textsf{tf}_4$ as recommended by Cremers et al. \cite{CPZ24}.
Specifically, this strengthens $\mathsf{tf_4}$ from requiring $S_2(lr) = S_4(lr)$ to $S_2(lr) = S_3(lr) = S_4(lr)$, ensuring that all honest signers in $lr.\mathcal{S}$ responded with honestly generated pre-processing tokens. 
The modification is purely definitional such that for echo schemes satisfying the standard pre-processing token check, it does not change the set of trivial forgeries. 
In particular, this holds for the \textsf{ATS[TS,DS]} transformation \cite{BTZ22,BCKMTZ22} whose partial signing only accepts signed pre-processing tokens, so the original and the strengthened $\mathsf{tf}_4$ coincide and the existing \textsf{TS-SUF-4} security proof for \textsf{ATS[TS,DS]} remains valid.

This strengthening approach, corresponds to Cremers et al.'s \textsf{SiGu-4:aLRhPP} level, restores linear downward implications across the hierarchy (e.g., \textsf{TS-(S)UF-4} implies \textsf{TS-(S)UF-3}) without weakening the lower levels as done in the revised version \cite{BTZ22}. 
We prefer this strengthened \textsf{TS-(S)UF-}$i$ model as shown in Figure \ref{fig:ts-uf-suf} over Cremers et al.'s full framework \cite{CPZ24} because it provides a standard comparison to \textsf{FROST} and its related variants. 

\begin{definition}
Let $\mathsf{X}\in \{\mathsf{suf,uf}\}$ and $i\in\{0,1,2,3,4\}$. A $\mathsf{TS}$ scheme is secure if there is no polynomial time algorithm $\mathcal{F}$ that can win the game $\mathbf{G}^{\mathsf{ts\text{-}X\text{-}}i}_{\mathsf{TS}}$ in time $T(\F{})$ with non-negligible advantage:
    \begin{align*}        \mathsf{Adv}^{\mathsf{ts\text{-}X\text{-}i}}_{\mathsf{TS}}(\F{})=\Pr[\mathbf{G}^{\mathsf{ts\text{-}X\text{-}}i}_{\mathsf{TS}}(\F{})=\mathsf{true}].
    \end{align*}
\end{definition}

\subsection{Generalized Forking Lemma Extension}
This lemma is proposed by Bellare et al. \cite{BTZ22,BCKMTZ22} in proving the \textsf{TS-SUF-2} security of \textsf{FROST2}.
It extends the generalized Forking Lemma \cite{BN06} where if $\mathcal{A}$ can only output index $I$ within a given set $SS \subseteq [1..q]$, then the final bound on $\mathsf{acc}(\mathsf{Fork}^{\mathcal{A}})$ depends only on $|S|$ instead of $q$. 
This reduces the security loss from $1/q$ to $1/|SS|$.

\begin{lemma}\label{lem:fork}
Let $q\ge 1$ be an integer, $SS\subseteq [1..q]$ be a set, and $\mathbb{H}$ be a set. Let $\mathcal{A}$ be a randomized algorithm that on input $x, h_1, \ldots, h_q$ outputs a pair $(i_{sig}, \mathsf{Out})$, where $i_{sig} \in \{\perp\} \cup SS$ and Out is a side output. Let \textsf{GrGen} be a randomized algorithm that generates $x$. The accepting probability of $\mathcal{A}$ is defined as
$\mathsf{acc}(\mathcal{A}) = \Pr_{x \gets \mathbf{GrGen},\, h_1,\ldots,h_q \gets \mathbb{H}}\bigl[(i_{sig},\mathsf{Out}) \gets \mathcal{A}(x,h_1,\ldots,h_q) : i_{sig} \ne \perp\bigr]$.

The accepting probability of $\mathsf{Fork}^{\mathcal{A}}$ is defined as
\[
\mathsf{acc}(\mathsf{Fork}^{\mathcal{A}}) = \Pr_{x \gets \mathbf{GrGen}}\bigl[\alpha \gets \mathsf{Fork}^{\mathcal{A}}(x) : \alpha \ne \perp\bigr].
\]
Then $\mathsf{acc}(\mathsf{Fork}^{\mathcal{A}}) \ge \mathsf{acc}(\mathcal{A})^2 / |SS|$.
\end{lemma}

\subsection{Generalized Multi-Instance Multi-Forking Lemma}
\label{sec:mimfork}
We extend the generalized multi-instance forking lemma of
Bagherzandi et al.~\cite{BCJ08} to allow multiple forks per
instance. 
The statement is average-case over parameters generated by $\mathsf{GrGen}$, in the same style as the original lemma.
In essence, we modify the success count of each local instance from $1$ to $n$, analogous to Gennaro et al.'s multi-reset lemma~\cite{GLSY04} built upon Bellare and Neven's generalized reset lemma~\cite{BN06}.

Let $\mathbb{H}$ be a finite set and let $q,N,n$ be positive integers. 
Consider an algorithm $\mathcal{A}$ that, on input
$\mathsf{par}$, internal coins $\rho$, and a random-oracle tape $h=(h_1,\dots,h_q)\in\mathbb{H}^q$, outputs an index set $I$ and side outputs $(\{\mathsf{out}_i\}_{i\in I},\theta)$. 
If $\mathcal{A}$ fails, it outputs $(I,\{\mathsf{out}_i\}_{i\in I},\theta)$ with $I=\emptyset$. Otherwise we assume $|I|=N$, treating $I$ as a multiset and allowing repetitions, and we assume that if $I=\{i_1,\dots,i_N\}$ then $i_1\le\cdots\le i_N$.
Note that we adopt the simplified output notations \cite{CGRS23} that has an extra $\theta$, which can be viewed as a suffix in every $\mathsf{out}_i$.

Let
$\Omega
=
\{(\rho,h_1,\dots,h_q)\mid \rho\in\{0,1\}^*,\ h_j\in\mathbb{H}\}$
be the set of all possible random tapes. For a fixed
$\mathsf{par}$, let $S(\mathsf{par})=\{f\in\Omega\mid \mathcal{A}(\mathsf{par},f)\text{ succeeds}\}$.
Define $\varepsilon(\mathsf{par}) = \Pr_{f\leftarrow\Omega}[\mathcal{A}(\mathsf{par},f)\text{ succeeds}]$ and let
$\hat\varepsilon
=
\Pr_{\mathsf{par}\leftarrow\mathsf{GrGen}(1^\lambda),\,
f\leftarrow\Omega}
[\mathcal{A}(\mathsf{par},f)\text{ succeeds}]$.

Let $\mathcal{O}=\{(h_{i_m},\mathsf{out}_{i_m})\}_{m=1}^N$ be the side outputs from the initial execution, and let
$\mathcal{O}'=\emptyset$ initially. The forking algorithm
$\mathsf{MIMFork}_{N,n}^{\mathcal{A}}(\mathsf{par})$
performs $N$ local rewinding processes, one for each occurrence
$i_m$ in the ordered multiset $I$. For each target occurrence
$i_m$, it resamples the random-oracle suffix starting at $i_m$
until it collects $n$ additional side outputs whose random-oracle
responses at $i_m$ are pairwise distinct and also distinct from the
original response $h_{i_m}$.

\begin{algorithm}[h]
\caption{Generalized Multi-Instance Multi-Forking Algorithm
$\mathsf{MIMFork}_{N,n}^{\mathcal{A}}(\mathsf{par})$}
\label{algo:mimfork}
\begin{algorithmic}[1]
\State Pick internal coins $\rho$ for $\mathcal{A}$ at random
\State $(h_1,\ldots,h_q)\xleftarrow{\$} \mathbb{H}^q$
\State $(I,\{\mathsf{out}_i\}_{i\in I},\theta)\leftarrow
\mathcal{A}(\mathsf{par},(h_1,\ldots,h_q);\rho)$
\If{$I=\emptyset$}
\State \Return $\bot$
\EndIf
\State Parse $I$ as a multiset
$\{i_1,i_2,\ldots,i_N\}$ such that
$i_1\le i_2\le \dots \le i_N$
\State
$\mathcal{O} \leftarrow \{(h_{i_m},\mathsf{out}_{i_m})\}_{m=1}^N~;~\mathcal{O}'\leftarrow\emptyset~;~L\leftarrow
\left\lceil
\frac{8Nq}{\hat\varepsilon}
\ln\left(\frac{8Nn}{\hat\varepsilon}\right)
\right\rceil ~;~ T_{\max}\leftarrow nL$
\For{$m=1$ \textbf{to} $N$}
\State $\mathcal{O}'_m\leftarrow\emptyset~;~c_m\leftarrow 0$
\While{$|\mathcal{O}'_m|<n$ \textbf{and} $c_m<T_{\max}$}
    \State $c_m\leftarrow c_m+1~;~h_{i_m}^{(m,c_m)},\ldots,h_q^{(m,c_m)}
    \xleftarrow{\$}
    \mathbb{H}^{q-i_m+1}$
    \State
    $Q^{(m,c_m)}
    \leftarrow
    \bigl(
    h_1,\ldots,h_{i_m-1},
    h_{i_m}^{(m,c_m)},\ldots,h_q^{(m,c_m)}
    \bigr)$
    \State
    $(I',\{\mathsf{out}'_i\}_{i\in I'},\theta')
    \leftarrow
    \mathcal{A}(\mathsf{par},Q^{(m,c_m)};\rho)$
    \If{$i_m\in I'$}
        \If{$h_{i_m}^{(m,c_m)}
            \notin
            \{h_{i_m}\}\cup
            \{h'\mid(h',\cdot)\in\mathcal{O}'_m\}$}
            \State
            $\mathcal{O}'_m
            \leftarrow
            \mathcal{O}'_m
            \cup
            \{(h_{i_m}^{(m,c_m)},\mathsf{out}'_{i_m})\}$
        \EndIf
    \EndIf
\EndWhile
\If{$|\mathcal{O}'_m|<n$}
    \State \Return $\bot$
\EndIf
\State $\mathcal{O}'\leftarrow\mathcal{O}'\cup\mathcal{O}'_m$
\EndFor
\State \Return $(\mathcal{O},\mathcal{O}')$
\end{algorithmic}
\end{algorithm}

\begin{lemma}[Generalized Multi-Instance Multi-Forking Lemma]
\label{lem:mimfork}
Let $\mathbb{H}$ be a finite set and let $q,N,n\in\mathbb{N}$ with $N,n\ge 1$. 
Let $\mathcal{A}$ be a randomized algorithm that, on input $\mathsf{par}$ and a vector
$h\in\mathbb{H}^q$, returns an index set $I$ and side outputs $(\{\mathsf{out}_i\}_{i\in I},\theta)$. 
Define $\hat\varepsilon$ as above.
Assume $|\mathbb{H}|>\frac{16Nnq}{\hat\varepsilon}$.
Then the accepting probability of $\mathsf{MIMFork}_{N,n}^{\mathcal{A}}$ is
\[
\mathsf{acc}
\left(
\mathsf{MIMFork}_{N,n}^{\mathcal{A}}
\right)
=
\Pr
\left[
\mathsf{MIMFork}_{N,n}^{\mathcal{A}}
(\mathsf{par})
\neq
\bot
\right]
\ge
\frac{\hat\varepsilon}{8},
\]
where the probability is over
$\mathsf{par}\leftarrow\mathsf{GrGen}(1^\lambda)$ and the coins of $\mathsf{MIMFork}_{N,n}^{\mathcal{A}}$.
Moreover, whenever the algorithm returns
$(\mathcal{O},\mathcal{O}')$, for every occurrence $i_m$ the collection $\mathcal{O}'_m$ contains $n$ pairs whose first components are pairwise distinct and distinct from $h_{i_m}$.
\end{lemma}

\begin{proof}
For a fixed $\mathsf{par}$, let
$\varepsilon=\varepsilon(\mathsf{par})$ and
$S=S(\mathsf{par})$ where $\Pr[S]=\varepsilon$ by definition.
For every $f\in S$, let
$I(f)=\{i_1(f),\ldots,i_N(f)\}$, where
$i_1(f)\le\cdots\le i_N(f)$.
For $j\in\{1,\ldots,q\}$, let
$f_j=(\rho,h_1,\ldots,h_{j-1})$ denote the prefix before the
$j$-th random-oracle response, and let
$s_j=(h_j,\ldots,h_q)$ denote the corresponding suffix.
For every $m\in\{1,\ldots,N\}$ and $j\in\{1,\ldots,q\}$, define
$A_m^j=\{f\in S\mid i_m(f)=j\}$ and
$B_m^j=
\left\{
f\in\Omega
\;\middle|\;
\Pr_{s_j}\bigl[(f_j\|s_j)\in A_m^j\bigr]
\ge
\frac{\varepsilon}{2Nq}
\right\}$.
Let
$E_m=\bigcup_{j=1}^q(A_m^j\cap B_m^j)$.
By part~(3) of the Splitting Lemma~\cite[Lemma 2]{BCJ08},
for every $m$,
\begin{align}
\Pr[E_m\mid S]
=\sum_{j=1}^q\Pr[A_m^j\cap B_m^j\mid S]\ge
\sum_{j=1}^q
\left(
\Pr[A_m^j\mid S]
-
\frac{\varepsilon/(2Nq)}{\Pr[S]}
\right)
=
1-\frac{1}{2N}.
\label{eq:good_occurrence}
\end{align}
Hence, by the union bound, it gives
\begin{equation}
\label{eq:initial_success}
\Pr\left[S\cap\bigcap_{m=1}^N E_m\right]
\ge
\varepsilon\left(1-\sum_{m=1}^N\Pr[\overline{E_m}\mid S]\right)
\ge
\frac{\varepsilon}{2}.
\end{equation}

Condition on $S\cap\bigcap_{m=1}^N E_m$.
Fix an occurrence $m$ and suppose that fewer than $n$ distinct
forks have been collected.
Since the original tape belongs to $E_m$, we have
$f\in B_m^{i_m}$, and hence a fresh suffix yields an execution in
$A_m^{i_m}$ with probability at least $\varepsilon/(2Nq)$.
Let
$H_{\mathsf{forbid}}
=
\{h_{i_m}\}\cup
\{h'\mid(h',\cdot)\in\mathcal{O}'_m\}$.
Since $|\mathcal{O}'_m|<n$, we have
$|H_{\mathsf{forbid}}|\le n$, and conditioned on the current
state of the local rewinding process, the probability that a
fresh trial produces an execution in $A_m^{i_m}$ and yields a
new response is at least
$p
\ge
\frac{\varepsilon}{2Nq}
-
\frac{n}{|\mathbb{H}|}
>
\frac{\varepsilon}{2Nq}
-
\frac{\hat\varepsilon}{16Nq}
>
\frac{\varepsilon}{4Nq}$,
where the second inequality uses
$|\mathbb{H}|>16Nnq/\hat\varepsilon$, and the last uses
$\varepsilon>\hat\varepsilon/2$ for the good parameters.

Thus, as long as fewer than $n$ forks have been collected, every
rewinding trial has conditional probability at least
$\varepsilon/(4Nq)$ of producing a new valid fork.
Recall that $T_{\max}=nL$, where now
$L=
\left\lceil
\frac{8Nq}{\hat\varepsilon}
\ln\left(\frac{8Nn}{\hat\varepsilon}\right)
\right\rceil$.
Partition the $nL$ trials into $n$ consecutive blocks of $L$
trials.
Since each trial has conditional success probability at least
$p>\varepsilon/(4Nq)>\hat\varepsilon/(8Nq)$, the probability
that a fixed block produces no new fork is at most
$(1-p)^L
\le
\exp(-pL)
\le
\frac{\hat\varepsilon}{8Nn}$.

Therefore, let $\mathsf{Fail}_m$ be the event that the $m$-th
occurrence fails to collect $n$ distinct forks.
By the union bound over the $n$ blocks and then over the $N$
occurrences,
\begin{equation}
\label{eq:global_failure}
\Pr\left[
\mathsf{Fail}_{\mathsf{any}}
\;\middle|\;
S\cap\bigcap_{m=1}^N E_m
\right]
\le
N\cdot n\cdot\frac{\hat\varepsilon}{8Nn}
=
\frac{\hat\varepsilon}{8},
\end{equation}
where
$\mathsf{Fail}_{\mathsf{any}}=\bigcup_{m=1}^N\mathsf{Fail}_m$.
Since we are conditioning on a good parameter with
$\varepsilon>\hat\varepsilon/2$, we have
$\hat\varepsilon/8<\varepsilon/4$.
Combining~\eqref{eq:initial_success}
and~\eqref{eq:global_failure}, for every
$\mathsf{par}$ with $\varepsilon(\mathsf{par})>\hat\varepsilon/2$,
we have
$\mathsf{frk}(\mathsf{par})
\ge
\frac{\varepsilon}{2}
-
\frac{\varepsilon}{4}
=
\frac{\varepsilon}{4}$.

Finally, average over
$\mathsf{par}\leftarrow\mathsf{GrGen}(1^\lambda)$.
Let
$P'
=
\left\{
\mathsf{par}
\mid
\varepsilon(\mathsf{par})>\hat\varepsilon/2
\right\}$.
For $\mathsf{par}\in P'$, we have
$\mathsf{frk}(\mathsf{par})
\ge
\frac{\varepsilon(\mathsf{par})}{4}$ while for $\mathsf{par}\notin P'$, we use the trivial bound
$\mathsf{frk}(\mathsf{par})\ge 0$.
Therefore,
\begin{align*}
\mathsf{acc}
\left(
\mathsf{MIMFork}_{N,n}^{\mathcal{A}}
\right)
&\ge
\frac{1}{|P|}
\sum_{\mathsf{par}\in P'}
\frac{\varepsilon(\mathsf{par})}{4} 
\\
&=
\frac{1}{4}
\left(
\hat\varepsilon
-
\frac{1}{|P|}
\sum_{\mathsf{par}\notin P'}
\varepsilon(\mathsf{par})
\right)
\ge
\frac{1}{4}
\left(
\hat\varepsilon
-
\frac{\hat\varepsilon}{2}
\right)
=
\frac{\hat\varepsilon}{8}.\qquad\qquad\qed
\end{align*}
\end{proof}

\section{Our Constructions} 
\label{sec:our_construction}

\begin{figure}
    \centering
    \begin{tabular}{|l l|}
    \hline
    \begin{minipage}{0.6\textwidth}
        \begin{tabular}{p{0.9\textwidth}}
            $\mathsf{Setup}(1^{\lambda})$ \\
            \hline
            $(\G, p, g)\leftarrow \mathsf{GrGen}(1^{\lambda})$\\
            $\mathsf{H}_{\textsf{pp}},\mathsf{H}_{\textsf{non}}, \mathsf{H}_{\textsf{sig}}:\{0,1\}^*\rightarrow \Zs{p}$\\
            $\textsf{par}\leftarrow ((\G, p, g),\mathsf{H}_{\textsf{pp}},\mathsf{H}_{\textsf{non}},\mathsf{H}_{\textsf{sig}})$\\
            return \textsf{par}\\
            \\
            $\mathsf{Kg}(\mathsf{par},t,\mathsf{ns}):$\\
            \hline
            For $i=0 \text{ to } t-1$, do\\
            $~~~a_i\in\Zs{p}$\\
            For $i=1 \text{ to }\mathsf{ns}$, do\\
            $~~~\mathsf{sk}_i\leftarrow \sum_{j=0}^{t-1}a_j\times \mathsf{id}_i^j;X_i\leftarrow g^{\mathsf{sk}_i}$\\
            $\tilde{X}=g^{a_0};aux\leftarrow(\{X_i\}_{1\leq i\leq \mathsf{ns}})$\\
            For $i=1 \text{ to }\mathsf{ns}$, do\\
            $~~~\mathsf{st}_i\leftarrow(\tilde{X},aux,\mathsf{sk}_i)$\\
            return $\{\mathsf{st}_i\}_{1\leq i\leq \mathsf{ns}}$\\
            \\
            $\mathsf{SPP}(\mathsf{st}_i)$\\
            \hline
            $r,s\boxed{,v} \xleftarrow{\$}\Zs{p}$\\ 
            $\boxed{e=\mathsf{H}_{\mathsf{pp}}(\tilde{X},X_i,R=g^r,S=g^s,V=g^v)}$\\
            $\boxed{u=v+\mathsf{st}_i.sk\cdot e+re^2+se^3}$\\
            $pp\leftarrow(R,S\boxed{,u,V})$\\
            $\mathsf{st}_i.\mathsf{mapPP}(pp) \leftarrow (r,s\boxed{,v})$\\
            
            return $(pp, \mathsf{st}_i)$\\
            \\
            $\mathsf{LPP}(\mathsf{id}_i, pp, \mathsf{st}_0)$:\\
            \hline
            $\mathsf{st}_0.\textsf{currPP}_i\leftarrow\mathsf{st}_0.\textsf{currPP}_i\cup\{pp\}$\\
            return $\mathsf{st}_0$\\
            \\
            $\mathsf{LR}(M, \mathcal{S},\mathsf{st}_0)$\\
            \hline
            If $\exists \mathsf{id}_i\in\mathcal{S}: \mathsf{st}_0.\textsf{curPP}(\mathsf{id}_i) = \emptyset$, return $(\perp, st_0)$\\
            $lr .\mathsf{msg} \leftarrow M;lr.\mathcal{S}\leftarrow \mathcal{S}$\\
            For $\mathsf{id}_i\in\mathcal{S}$, do\\
            $~~~$Pick $pp_i$ from $\mathsf{st}_0.\textsf{curPP}(\mathsf{id}_i)$\\
            $~~~lr.\mathsf{PP}(\mathsf{id}_i)\leftarrow pp_i$\\
            $~~~\mathsf{st}_0.\mathsf{curPP}(\mathsf{id}_i) \leftarrow \mathsf{st}_0.\mathsf{curPP}(\mathsf{id}_i) \setminus \{pp_i\}$\\            
            return $(lr, st_0)$\\
            \\
            $\mathsf{Vf}(\tilde{X}, M, \sigma)$\\
            \hline
            $(\tilde{R}, z) \leftarrow \sigma$ ; $c \leftarrow \mathsf{H_{sig}}(\tilde{X}, M, \tilde{R})$\\
            return $(g^z=\tilde{R}\tilde{X}^c)$\\
        \end{tabular}
    \end{minipage}
         & 
    \begin{minipage}{0.5\textwidth}
        \begin{tabular}{p{0.9\textwidth}}
              $\mathsf{CompPar}(\tilde{X},lr)$\\
              \hline
              $M\leftarrow lr.\textsf{msg}$\\
              $d \leftarrow \mathsf{H}_{\mathsf{non}}(\tilde{X}, lr)$\\
              $(R_i,S_i)\leftarrow lr.\textsf{PP}(\mathsf{id}_i)$\\
              $\tilde{R} \leftarrow \prod_{\mathsf{id}_i\in lr.\mathcal{S}} R_i (\prod_{\mathsf{id}_i\in lr.\mathcal{S}} S_i)^d$\\        
              $c \leftarrow \mathsf{H}_{\textsf{sig}}(\tilde{X}, M, \tilde{R})$\\
              return $(\tilde{R},c,d)$\\
              \\
              $\mathsf{PS}(lr,\mathsf{st}_i)$\\
              \hline             
              $pp_i\leftarrow lr.\textsf{PP}(\mathsf{id}_i)$\\
              If $\mathsf{st}_i.\mathsf{mapPP}(pp_i)=\perp$, return $(\perp,\mathsf{st}_i)$\\
              $\boxed{\tilde{X},aux\leftarrow \mathsf{st}_i}$\\
              $\boxed{\text{For }\mathsf{id}_i\in lr.\mathcal{S}\text{, do}}$\\
              $~~~\boxed{(R_i,S_i,u_i,V_i)\leftarrow lr.\mathsf{PP}(\mathsf{id}_i)}$\\
              $~~~\boxed{e'=\mathsf{H}_{\mathsf{pp}}(\tilde{X},X_i,R_i,S_i,V_i)}$\\
              $~~~\boxed{\text{If }g^{u_i}\neq V_i X_i^{e'}R_i^{e'^2}S_i^{e'^3}\text{, return}\perp}$\\
              $~~~\boxed{\text{Else }lr.\mathsf{PP}(\mathsf{id}_i)\leftarrow (R_i,S_i)}$\\
              $(r_i, s_i)\leftarrow \mathsf{st}_i.\mathsf{mapPP}(pp_i)$\\
              $\mathsf{st}_i.\mathsf{mapPP}(pp_i)\leftarrow \perp$\\              
              $(\tilde{R},c,d)\leftarrow \mathsf{CompPar}(st_i.\tilde{X},lr)$\\
              $z_i\leftarrow r_i + d\cdot s_i + c\cdot \Delta_{\mathsf{id}_i,\mathcal{S}}(0)\cdot \mathsf{st}_i.\mathsf{sk}$\\
              $psig_i\leftarrow (\tilde{R},z_i)$\\
              return $(psig_i, \mathsf{st}_i)$\\
              \\
              $\mathsf{Agg}(\{psig_i\}_{\mathsf{id}_i\in lr.\mathcal{S}} ,\mathsf{st}_0)$\\
              \hline
              $\tilde{R}\leftarrow \perp; z\leftarrow 0$\\
              For $\mathsf{id}_i\in\mathcal{S}$, do\\
              $~~~$If $\tilde{R}=\perp$, $\tilde{R}\leftarrow psig_i.\tilde{R}$\\
              $~~~$If $\tilde{R}\neq psig_i.\tilde{R}$, return $(\perp,\mathsf{st}_0)$\\
             $~~~z \leftarrow z+psig_i.z$\\
             $\sigma \leftarrow (\tilde{R}, z )$\\
             return $(\sigma,\mathsf{st}_0)$\\
            \\
            $\mathsf{SVf}(\tilde{X}, lr, \sigma)$\\
            \hline
            $(\tilde{R}, z) \leftarrow \sigma$\\
            $(\tilde{R}',c',d) \leftarrow \mathsf{CompPar}(\tilde{X}, lr)$\\
            return $(\tilde{R}=\tilde{R}')\wedge(g^z=\tilde{R}\tilde{X}^{c'})$
        \end{tabular}
    \end{minipage} \\ 
     & \\ 
     \hline
    \end{tabular}
    \caption{The \textsf{FROST2} and $\boxed{\textsf{FROST2}\texttt{+}}$ signature schemes where \textsf{FROST2} contains all but the solid box. \textsf{GrGen} generates cyclic group parameter in which the discrete logarithm assumption holds, $\mathsf{st}$ is state, $\mathsf{ns}<p-1$ is the number of signers, and $t$ is the signing threshold. The function \textsf{mapPP} records dlogs of tokens $pp$, and \textsf{currPP} is the current $pp$ collection.}
    \label{fig:our_frost2}
\end{figure}

In this section we present \textsf{FROST2}\texttt{+} (shown in Figure~\ref{fig:our_frost2}) as a \textsf{TS-SUF-4} variant of \textsf{FROST2} from Bellare et al.~\cite{BCKMTZ22}.
Different from Bellare et al.'s \textsf{ATS[TS, DS]} transformation~\cite{BCKMTZ22}, which requires each signer to manage an extra digital signature key-pair to authenticate the pre-processing token $(R=g^r,S=g^s)$, ours reuses only the existing signing key share and the secret exponents already presented in the token.
That said, \textsf{FROST2}\texttt{+} has the same construction as \textsf{FROST2} except in the \textsf{SPP} and \textsf{PS} algorithms.
The \textsf{SPP} now broadcasts not only the pre-processing token $(R,S)$ but also a short authentication transcript $(u,V=g^v)$ with:
\[
u=v+ \mathsf{st}_i.\mathsf{sk}\cdot e+re^2+se^3
\]
which is to be verified during \textsf{PS}.
A signer aborts \textsf{PS} whenever a peer's $(u,V)$ verification fails.
The correctness follows directly from the equations such that if all parties are honest, then for each authentication transcript $(u,V)$, we have $g^u=VX_i^eR^{e^2}S^{e^3}$. In the signing phase, the partial signatures satisfy $z_i=r_i+ds_i+c\Delta_{\mathsf{id}_i,\mathcal{S}}(0)\mathsf{st}_i.\mathsf{sk}$. Aggregating $z=\sum_{\mathsf{id}_i\in\mathcal{S}}z_i$ gives $g^z=\tilde{R}\tilde{X}^c$ and \textsf{Vf} returns 1.
Moreover, as discussed in the subsequent sections, when considering a batch verification of $k$ tokens $(R_1,S_1,\ldots,R_k,S_k)$, the transcript size remains constant.

\subsection{Pre-Computations}
\label{sec:pre_comp}
Although the \textsf{TS-SUF-}$i$ notions of Bellare et al.~\cite{BCKMTZ22} assume a designated leader for collecting and distributing pre-processing tokens and signing requests to simplify security analysis, they explicitly cover scenarios in which corrupted signers take turns acting as the leader. Consequently, a leaderless design introduces no unmodeled attack vectors. 
In particular, attacks captured by \textsf{TS-UF-3} and \textsf{TS-UF-4}, where a corrupted leader forging tokens or requests, are already included in the \textsf{TS-SUF-}$i$ adversary model. 
When an honest signer initiates the protocol, the same corruption bounds apply, preserving the original security guarantees.

If our proposed \textsf{FROST2}\texttt{+} adopts this leaderless approach, signers can collect and verify pre-processing tokens as in the original \textsf{FROST}~\cite{FROST} prior to the partial signing (\textsf{PS}) phase.
This pre-verification step improves runtime performance, as shown in Table~\ref{tab:TS_comparison}, without compromising security.

\subsection{Batch Verification}
\label{sec:batch_verif}
Similar to \textsf{ATS[FROST,DS]} \cite{BTZ22,BCKMTZ22} which can generate an offline signature on multiple tokens for batch verification, our \textsf{FROST2}\texttt{+} generates an offline authentication transcript on multiple tokens for the same purpose.
Specifically, when a signer $\textsf{id}_i$ has $n(\geq 2)$ tokens to be verified, during \textsf{SPP}, the signer computes:
\begin{align*}
    &\{r_{i,j},s_{i,j}\}_{j=1}^n,v \xleftarrow{\$}\Zs{p},~e=\mathsf{H}_{\mathsf{pp}}(\tilde{X},X_i,\{R_{i,j},S_{i,j}\}_{j=1}^n,V=g^v)\\          
    &u=v+\mathsf{st}_i.\mathsf{sk}\cdot e+\sum_{j=1}^nr_{i,j}e^{j+1}+\sum_{j=1}^ns_{i,j}e^{j+1+n},~            pp\leftarrow(\{R_{i,j},S_{i,j}\}_{j=1}^n,u,V)
\end{align*}
where the transcript size $(u,V)$ remains constant.
Subsequently, the verification is done by checking whether:
$$g^{u}= VX_i^{e'}\prod_{j=1}^nR_{i,j}^{e'^{j+1}}\prod_{j=1}^nS_{i,j}^{e'^{j+1+n}}$$
holds where $e'=\mathsf{H}_{\mathsf{pp}}(\tilde{X},X_i,\{R_{i,j},S_{i,j}\}_{j=1}^n,V)$.

Enabling this batch verification feature does not affect the \textsf{TS-SUF-4} security level, but the reduction loss because the run time of reduction is linear to $n$ as shown in Section \ref{sec:security_analysis}. 
If a tighter reduction is desirable, then one may need to explore the Multi-Instance Reset Lemma \cite{KMP16}.

\section{Security Analysis}
\label{sec:security_analysis}
We prove the \textsf{TS-SUF-4} security of \textsf{FROST2}\texttt{+} under the OMDL assumption in the random oracle model. 
Our proving strategy is to separate the security analysis into three parts which handle the events
\begin{enumerate}
    \item $E_1:\mathsf{Vf}(\tilde{X}, M^*, \sigma^*) ~\wedge~ (lr^* = \perp ~\vee~ |S_2(lr^*)| < t - |CS|)$,
    \item $E_2:\mathsf{Vf}(\tilde{X}, M^*, \sigma^*) ~\wedge~ lr^* \ne \perp ~\wedge~  S_2(lr^*)\ne S_3(lr^*)$,
    \item $E_3:\mathsf{Vf}(\tilde{X}, M^*, \sigma^*) ~\wedge~ lr^* \ne \perp ~\wedge ~ \neg(S_2(lr^*)= S_3(lr^*) = S_4(lr^*))$,
\end{enumerate}
respectively.
We first present the security proofs in the centralized key generation setting to ease exposition, and show how to extend them to the decentralized key generation setting in Section \ref{sec:dkg_proof}.

\begin{theorem}
\label{thm:frost2+}
Let $\mathcal{F}$ be a \textsf{TS-SUF-4} forger against $\mathsf{FROST2+}$ making at most $q$ queries to the random oracle. Suppose the success probability of $\mathcal{F}$ is bounded by $\mathsf{Adv}_{\mathsf{FROST2+}}^{\mathsf{ts\text{-}suf\text{-}4}}(\mathcal{F}) \le \Pr[E_1] + \Pr[E_2] + \Pr[E_3]$. Then, there exist three OMDL adversaries $\mathcal{B}_1,\mathcal{B}_2$ and $\mathcal{B}_3$ such that:
\begin{align*}
\mathsf{Adv}_{\mathsf{FROST2+}}^{\mathsf{ts\text{-}suf\text{-}4}}(\mathcal{F}) 
\le \sqrt{q \cdot \mathsf{Adv}_{\G}^{\mathsf{omdl}}(\mathcal{B}_1)} + 8 \mathsf{Adv}_{\mathbb{G}}^{\mathsf{omdl}}(\mathcal{B}_2) + 8\mathsf{Adv}_{\mathbb{G}}^{\mathsf{omdl}}(\mathcal{B}_3) + \mathsf{negl}
\end{align*}
where $\mathcal{B}_1$ and $\mathcal{B}_3$ issue at most $2q_s+t$ queries to the \textsc{Chall} oracle, while $\mathcal{B}_2$ issues at most $2q_s$ queries. 
Moreover, their running times are bounded above by
\begin{align*}
T(\mathcal{B}_1) &\le 2 T(\mathcal{F}) + T(q),\\
T(\mathcal{B}_2) &\le \frac{16\mathsf{ns}^2q \ln(16\mathsf{ns}/\Pr[E_2])}{\Pr[E_2]} T(\F{})+T(q),\text{ and}\\ 
T(\mathcal{B}_3) &\le \frac{24 q \ln(24/\Pr[E_3])}{\Pr[E_3]}\cdot T(\F{}) +T(q), 
\end{align*}
respectively,
where $t\le\mathsf{ns}$, and $T(q)$ denotes the time required for the relevant group operations and exponentiations.
\end{theorem}

\begin{proof}
The proof follows from Lemmas \ref{lem:tssuf4_1}, \ref{lem:tssuf4_2} and \ref{lem:tssuf4_3}. \qed
\end{proof}

Lemma \ref{lem:tssuf4_1} is intended for the event $E_1$, where the forgery must contain at least one forged $psig_i$. The proof is similar to the \textsf{TS-SUF-2} proof of \textsf{FROST2}~\cite[Lemma 5.2]{BTZ22}, except that we add an additional $\textsc{RO}[\mathsf{H_{pp}}]$ to handle the authentication transcript $(u,V)$ and adjust the sequence of hash queries accordingly.


\begin{lemma}
\label{lem:tssuf4_1}
    There exists an OMDL adversary $\mathcal{B}_1$ making at most $2q_s+t$ queries to the \textsc{Chall} oracle and $q$ queries to the random oracle such that 
    $$\Pr[E_1]\le \sqrt{q \cdot \mathsf{Adv}_{\G}^{\mathsf{omdl}}(\mathcal{B}_1) + \mathsf{negl}}.$$
    Moreover, $\mathcal{B}_1$ runs in time roughly twice that of \F{}, plus the time $T(q)$ required to perform exponentiations and group operations.
\end{lemma}
    
\begin{proof}
Let \F{} be an adversary as described in the theorem. Denote the output message-signature pair of \F{} as $(M^*, sig^* = (\tilde{R}^*, z^*))$. Without loss of generality, we assume \F{} always queries \textsc{RO} on $\mathsf{H_{sig}}(\tilde{X}, M^*, \tilde{R}^*)$ before \F{} returns and always queries \textsc{RO} on $\mathsf{H_{pp}}(\tilde{X}, X_i, R, S, V)$ and $\mathsf{H_{non}}(\tilde{X}, lr)$ prior to the query $\textsc{PSignO}(\mathsf{id}_i, lr)$ for some $\mathsf{id}_i$ and $lr$.
This adds up to $2q_s$ additional \textsc{RO} queries, and we let $q = q_h + 2q_s + 1$ where $q_h=q_{pp} + q_{non} + q_{sig}$.

Let $lr^*$ be the leader query whose nonce query $\mathsf{H_{non}}(\tilde{X}, lr^*)$ is the last such query before the signature hash query $\mathsf{H_{sig}}(\tilde{X}, M^*, \tilde{R}^*)$ and which satisfies $\mathsf{SVf}(\tilde{X},$ $ lr^*, \sigma^*) = \mathsf{true}$ (if none exist, $lr^*=\perp$). Denote the event $E_1$ as
\begin{align*}
\mathsf{Vf}(\tilde{X}, M^*, \sigma^*) \wedge (&lr^* = \perp \vee~ |S_2(lr^*)| < t - |CS|).    
\end{align*}

We first construct an algorithm \C compatible with the syntax in Lemma~\ref{lem:fork}. The input of \C consists of $2q_s + t$ uniformly random group elements $\mathbf{A}=\{A_0, \ldots, A_{t-1}\}, \mathbf{R}=\{R_1, \ldots, R_{q_s}\}, \mathbf{S}=\{S_1,\ldots,S_{q_s}\} \in \G$ and uniformly random integers $h_1,\ldots, h_{2(q-q_{pp})},h_{2(q-q_{pp})+1},\ldots,h_{2q} \in \mathbb{Z}_p$. Also, \C can access an oracle \textsc{Dlog}, which on input $X \in \G$ outputs $\mathsf{DL}_{\G,g}(X)$. 
To start with, \C initializes all the states $\mathsf{st}_0, \ldots, \mathsf{st}_{\mathsf{ns}}$. 
In addition, it initializes counters $\mathsf{ctr}_s, \mathsf{ctr}_h$ to 0, $\mathsf{ctr}_{pp}$ to $2(q-q_{pp})$ and a function \textsf{dt} to an empty table, which are used to record the \textsc{Dlog} query related to each $(R_j, S_j)$. \C also initializes $\mathsf{curLR} \leftarrow \emptyset$ to record all leader requests that appear during the game and initializes $\mathsf{ctrPP}$ to an empty table, which are used to record the counter corresponding to each token generated by honest signers. 
We also use a flag \textsf{BadPPO} to denote whether a bad event occurs, which is initially set to \textsf{false}. Then, \C runs \F{} with access to the oracles $\widetilde{\textsc{Init}}, \widetilde{\textsc{PPO}}, \widetilde{\textsc{PSignO}},\widetilde{\textsc{RO}}$, which are simulated as follows.

\medskip
\noindent $\widetilde{\textsc{Init}}(CS)$: \C initializes $\mathsf{H}$ to an empty table and sets
$\tilde{X} \leftarrow A_0,X_i = \prod_{j=0}^{t-1} A_j^{\mathsf{id}_i^{\,j}}$ for $i \in [1..\mathsf{ns}]$,
and $\mathsf{st}_i.\mathsf{sk} = \textsc{Dlog}(X_i)$ for $\mathsf{id}_i \in CS$. Finally, \C returns $\tilde{X}$, $\mathsf{aux} = (X_1, \ldots, X_{\mathsf{ns}})$, $\{\mathsf{st}_i.\mathsf{sk}\}_{\mathsf{id}_i \in CS}$.

\medskip
\noindent $\widetilde{\textsc{RO}}[\mathsf{H_{pp}}](x)$: 
If $\mathsf{H_{pp}}(x)\neq \perp$, \C returns $\mathsf{H_{pp}}(x)$. 
Otherwise, \C parses $x$ as $(\tilde{X},X_i,R,S,V)$. 
If the parsing fails or $\tilde{X}\neq A_0$, \C sets $\mathsf{H_{pp}}(x)\xleftarrow{\$}\Zs{p}$ and returns $\mathsf{H_{pp}}(x)$.
Otherwise, \C sets $\mathsf{ctr}_{pp}\leftarrow\mathsf{ctr}_{pp} +1$ and $\mathsf{H_{pp}}(x)\gets h_{\mathsf{ctr}_{pp}}$ then returns $\mathsf{H_{pp}}(x)$.

\medskip
\noindent $\widetilde{\textsc{RO}}[\mathsf{H_{non}}](x)$: If $\mathsf{H_{non}}(x) \ne \perp$, \C returns $\mathsf{H_{non}}(x)$. Otherwise, \C parses $x$ as $(\tilde{X}, lr)$. If the parsing fails or $\tilde{X} \ne A_0$, \C sets $\mathsf{H_{non}}(x) \gets \Zs{p}$ and returns $\mathsf{H_{non}}(x)$. Else, \C sets $\mathsf{ctr}_h\leftarrow\mathsf{ctr}_h +1$ and $\mathsf{H_{non}}(x) \leftarrow h_{2\mathsf{ctr}_h - 1}$. It then computes
$\tilde{R} \leftarrow \prod_{\mathsf{id}_i \in lr.\mathcal{S}} R_i S_i^{\,h_{2\mathsf{ctr}_h - 1}}$,
where $(R_i, S_i, u_i, V_i) \gets lr.\mathsf{PP}(\mathsf{id}_i)$. 
If $\mathsf{H_{sig}}(\tilde{X}, lr.\mathsf{msg}, \tilde{R}) = \perp$, \C sets $\mathsf{H_{sig}}(\tilde{X}, lr.\mathsf{msg}, \tilde{R}) = h_{2\mathsf{ctr}_h}$. 
In addition, \C defines $\mathsf{mapLR}(\mathsf{ctr}_h) = lr$ and sets $\mathsf{curLR} \leftarrow \mathsf{curLR} \cup \{lr\}$. Finally, \C returns $\mathsf{H_{non}}(x)$.

\medskip
\noindent $\widetilde{\textsc{RO}}[\mathsf{H_{sig}}](x)$: If $\mathsf{H_{sig}}(x) \ne \perp$, \C returns $\mathsf{H_{sig}}(x)$. Otherwise, \C parses $x$ as $(\widetilde{X}, M, \tilde{R})$. If the parsing fails or $\widetilde{X} \ne X$, \C sets $\mathsf{H_{sig}}(x) \xleftarrow{\$} \Zs{p}$ and returns $\mathsf{H_{sig}}(x)$. Else, \C sets $\mathsf{ctr}_h\leftarrow\mathsf{ctr}_h+1$ and $\mathsf{H_{sig}}(x) \leftarrow h_{2\mathsf{ctr}_h}$. Finally, \C returns $\mathsf{H_{sig}}(x)$.

\medskip
\noindent $\widetilde{\textsc{PPO}}(\mathsf{id}_i)$ query: 
Same as in the game $\mathbf{G}_{\mathsf{FROST2}\texttt{+}}^{\textsf{ts-suf-4}}$, except that, in the simulation of \textsf{SPP}, \C increments $\mathsf{ctr}_s$ and $\mathsf{ctr}_{pp}$ by 1, chooses $u_{\mathsf{ctr}_{pp}}\xleftarrow{\$}\Zs{p}$ and computes
$V_{\mathsf{ctr}_{pp}} = g^{u_{\mathsf{ctr}_{pp}}}/X_i^{h_{\mathsf{ctr}_{pp}}}R_{\mathsf{ctr}_s}^{h_{\mathsf{ctr}_{pp}}^2}S_{\mathsf{ctr}_s}^{h_{\mathsf{ctr}_{pp}}^3}$.
If $\mathsf{H_{pp}}(\tilde{X},X_i,R_{\mathsf{ctr}_s}, S_{\mathsf{ctr}_s}, V_{\mathsf{ctr}_h})=\perp$, \C sets $\mathsf{H_{pp}}(\tilde{X},X_i,R_{\mathsf{ctr}_s}, S_{\mathsf{ctr}_s}, V_{\mathsf{ctr}_h})\gets h_{\mathsf{ctr}_{pp}}$, $pp \leftarrow (R_{\mathsf{ctr}_s}, S_{\mathsf{ctr}_s}, u_{\mathsf{ctr}_h}, V_{\mathsf{ctr}_h})$, $\mathsf{st}_i.\mathsf{mapPP}(pp) \leftarrow (0,0,0,0)$, and $\mathsf{ctrPP}(\mathsf{id}_i, pp) \leftarrow \mathsf{ctr}_s$. 
Else if $\mathsf{H_{pp}}(\tilde{X},X_i,R_{\mathsf{ctr}_s}, S_{\mathsf{ctr}_s}, V_{\mathsf{ctr}_{pp}})\ne\perp$, \textsf{BadPPO} is set to \textsf{true}.
If there exists $lr \in \mathsf{curLR}$ such that $lr.\mathsf{PP}(\mathsf{id}_i) = (R_{\mathsf{ctr}_s}, S_{\mathsf{ctr}_s}, u_{\mathsf{ctr}_{pp}}', V_{\mathsf{ctr}_{pp}}')$ for some valid $(u_{\mathsf{ctr}_{pp}}'\neq u_{\mathsf{ctr}_{pp}},V_{\mathsf{ctr}_{pp}}'\neq V_{\mathsf{ctr}_{pp}})$, then \textsf{BadPPO} is also set to \textsf{true}.

\medskip
\noindent $\widetilde{\textsc{PSignO}}(\mathsf{id}_i, lr)$ query: Same as in the game $\mathbf{G}_{\mathsf{FROST2}\texttt{+}}^{\textsf{ts-suf-4}}$, except that, in the simulation of \textsf{PS}, if $\mathsf{st}_i.\mathsf{mapPP}(pp) \ne \perp$, \C sets
$z_i \gets \textsc{Dlog}\bigl( R_j S_j^{d_i} X_i^{\,c \Delta_{\mathsf{id}_i,lr.\mathcal{S}}(0)} \bigr)$,
where $j \gets \mathsf{ctrPP}(\mathsf{id}_i, lr.\mathsf{PP}(\mathsf{id}_i))$. In addition, \C sets $\mathsf{dt}(j) \gets (\mathsf{id}_i, k, d_i, c\Delta_{\mathsf{id}_i,lr.\mathcal{S}}(0), z_i)$, where $k$ denotes the index such that $\mathsf{H_{non}}(\tilde{X}, lr)$ is set to $h_{2k-1}$ during the simulation.\\

After receiving the output $(lr^*=\{M^*,\mathcal{S},\mathsf{PP}\}, \sigma^* = (\tilde{R}^*, z^*))$ from \F{}, \C returns $\perp$ if $\textsf{BadPPO}=\textsf{true}$, or $E_1$ does not occur. 
Else, \C finds the index $i_{sig}$ such that $\mathsf{H_{sig}}(\tilde{X}, M^*, \tilde{R}^*)$ is set to $h_{i_{sig}}$ during the simulation. By our assumption on \F{}, we know such $i_{sig}$ must exist. 
Then, \C returns $(i_{sig}, \mathsf{Out})$, where \textsf{Out} consists of all variables received or generated by \C.

\medskip
\noindent \textit{Analysis of \C.} To apply Lemma~\ref{lem:fork}, we define $SS := \{2k\}_{k \in [1..q-q_{pp}]}$ and let \textsf{GrGen} be the algorithm that samples $2q_s + t$ group elements uniformly from $\G$ and outputs them. From the simulation, the output index $i_{sig}$ of \C is always in $SS$. Moreover, it is clear that \C simulates the game $\mathbf{G}_{\mathsf{FROST2}\texttt{+}}^{\textsf{ts-suf-4}}$ perfectly when all the inputs to \C are uniformly sampled from their domain, which implies $\mathsf{acc}(\mathcal{C}) \ge \Pr[E_1] - \Pr[\mathsf{BadPPO}]$, where $\Pr[E_1]$ denotes the probability in the original $\mathbf{G}_{\mathsf{FROST2}\texttt{+}}^{\textsf{ts-suf-4}}$ game with \F{} as in Lemma~\ref{lem:fork}, and $\Pr[\mathsf{BadPPO}]$ denotes the probability that $\textsf{BadPPO}=\textsf{true}$ at the end of \C's execution. 
Since each pair $(R_j, S_j)$ is sampled uniformly from $\G$ and the corresponding authentication transcript $(u,V)$ is generated by programming $\textsc{RO}[\mathsf{H_{pp}}]$, for each $\textsc{PPO}(\mathsf{id}_i)$ query, the probability \textsf{BadPPO} is set to \textsf{true} is negligible and at most $q_h / p + |\mathsf{curLR}|/p=2q_h/p$. 
Therefore, we have $\Pr[\mathsf{BadPPO}] \le 2q_s q_h / p$. By Lemma~\ref{lem:fork},
\begin{align*}
\mathsf{acc}(\mathsf{Fork}^{\mathcal{C}}) \ge \frac{(\Pr[E_1] - 2q_s q_h / p)^2}{q-q_{pp}} 
\ge \frac{\Pr[E_1]^2}{q} - \frac{4\Pr[E_1] q_s q_h}{p \cdot q}
\ge \frac{\Pr[E_1]^2}{q} - \frac{2q_h}{p}.    
\end{align*}

\medskip
\noindent \textit{Construct $\mathcal{B}_1$ from $\mathsf{Fork}^{\mathcal{C}}$.} We now construct the OMDL adversary $\mathcal{B}_1$ using $\mathsf{Fork}^{\mathcal{C}}$ and the available \textsc{Dlog} oracle. To begin, $\mathcal{B}_1$ queries \textsc{Init} and then queries \textsc{Chall} oracle $2q_s + t$ times to generate the three lists of group elements $(\mathbf{A},\mathbf{R},\mathbf{S})$ as the input to run $\mathsf{Fork}^{\mathcal{C}}$. 
Without loss of generality, we can assume all the OMDL challenges are distinct, since otherwise $\mathcal{B}_1$ can solve them trivially. 
All \textsc{Dlog} queries from $\mathsf{Fork}^{\mathcal{C}}$ are forwarded by $\mathcal{B}_1$ to the \textsc{Dlog} oracle in the game $\mathbf{G}_{\G}^{\mathsf{omdl}}$. Denote by \textsf{BadHash} the event that any two of the scalars $h_{1}, h_{1}', \ldots, h_{q}, h_{q}'$ generated in the execution of $\mathsf{Fork}^{\mathcal{C}}$ are equal. 
Since these scalars are sampled uniformly from $\Zs{p}$, we have $\Pr[\mathsf{BadHash}] \le 2q^2 / p$.

It remains to show that if $\mathsf{Fork}^{\mathcal{C}}$ returns $(i_{sig}, \mathsf{Out}, \mathsf{Out}')$ and \textsf{BadHash} does not occur, $\mathcal{B}_1$ wins the game $\mathbf{G}_{\G}^{\mathsf{omdl}}$.
Absorbing all strictly negligible terms into a function $\mathsf{negl}$, we have
\[
\mathsf{Adv}_{\mathbb{G}}^{\mathsf{omdl}}(\mathcal{B}_1) \ge \mathsf{acc}(\mathsf{Fork}^{\mathcal{C}}) - \Pr[\mathsf{BadHash}] \ge \frac{\Pr[E_1]^2}{q} - \frac{3q^2}{p}= \frac{\Pr[E_1]^2}{q} - \mathsf{negl}.
\]

We use the notations introduced in the description of \C to denote the variables in \textsf{Out} and append a prime $(\cdot)'$ to denote the variables in $\mathsf{Out}'$. By the execution of $\mathsf{Fork}^{\mathcal{C}}$, we have $(\tilde{X}, M^*, \tilde{R}^*) = (\tilde{X}', M^{*'}, \tilde{R}^{*'})$ and $\tilde{X} = A_0$. Since $i_{sig} \in SS$, let $k^* = i_{sig} / 2$ where $\mathsf{mapLR}(k^*) = lr^*$. Furthermore, if $\mathsf{mapLR}(k^*) = \perp$, then $lr^*=\perp$ as well.

We first show how to compute the discrete logarithm of $A_0, \ldots, A_{t-1}$. Denote these discrete logarithms by $a_0, \ldots, a_{t-1}$ and define $f(x) := \sum_{i=0}^{t-1} a_i x^i$. Since \textsf{BadHash} does not occur, we have $\mathsf{H_{sig}}(\tilde{X}, M^*, \tilde{R}^*) = h_{i_{sig}} \ne h_{i_{sig}}' = \mathsf{H_{sig}}'(\tilde{X}, M^*, \tilde{R}^*)$. As $g^{z^*} = \tilde{R}^* g^{a_0 h_{i_{sig}}}$ and $g^{z^{*'}} = \tilde{R}^{*} g^{a_0 h_{i_{sig}}'}$, $\mathcal{B}_1$ computes $f(0) = a_0 = \frac{z^* - z^{*'}}{h_{i_{sig}} - h_{i_{sig}}'}$. Define $T_{\mathsf{dt}} := \{j : (i,k,d,c,z) \leftarrow \mathsf{dt}(j),\; k = k^*\}$. For each $j \in T_{\mathsf{dt}} \cap T_{\mathsf{dt}'}$, let $(i,k,d,c,z) \leftarrow \mathsf{dt}(j)$ and $(i',k',d',c',z') \leftarrow \mathsf{dt}'(j)$, then
$g^z = R_j S_j^d X_i^c$ and $g^{z'} = R_j S_j^{d'} X_i^{c'}$.

Since $\textsf{BadPPO}=\textsf{false}$ during both executions of \C, we know that $(R_j, S_j, u_j, V_j)$ was returned by a query $\textsc{PPO}(\mathsf{id}_i)$ prior to the query $\mathsf{H_{sig}}(\tilde{X}, M^*, \tilde{R}^*)$ in the first execution of \C.
Because the two executions of \C are identical up to the query $\mathsf{H_{sig}}(\tilde{X}, M^*, \tilde{R}^*)$, we have $\mathsf{id}_i' = \mathsf{id}_i$. 
Moreover, we know $d = h_{k} = h_{k^*} = h_{k'} = d'$. Therefore, $\mathcal{B}_1$ can compute
$f(\mathsf{id}_i) = \mathsf{DL}_{\G,g}(X_i) = \frac{z - z'}{c - c'}$.

Let $D := \{\mathsf{id}_i\}_{j \in T_{\mathsf{dt}} \cap T_{\mathsf{dt}'},\; (i,k,d,c,z) \leftarrow \mathsf{dt}(j)}$. 
If $S_2(lr^*)< t-|CS|$ occurs in the first execution of \C, then $|T_{\mathsf{dt}}| = |S_2(lr^*)| < t - |CS|$. Therefore, we have $|D| = |T_{\mathsf{dt}} \cap T_{\mathsf{dt}'}| < t - |CS|$. Consequently, $\mathcal{B}_1$ can pick an arbitrary set $D' \subseteq HS \setminus D$ of size $(t - |CS| - |T_{\mathsf{dt}} \cap T_{\mathsf{dt}'}| - 1)$ and for each $\mathsf{id}_i \in D'$, $\mathcal{B}_1$ queries the \textsc{Dlog} oracle on $X_i$. Thus $\mathcal{B}_1$ knows the value of $f(\mathsf{id}_i)$ for every $\mathsf{id}_i \in CS \cup D \cup D' \cup \{0\}$. Since $|CS \cup D \cup D' \cup \{0\}| = t$, $\mathcal{B}_1$ can recover $a_0, \ldots, a_{t-1}$ by Lagrange interpolation.

We now show how to compute the discrete log of $R_1, S_1, \ldots, R_{q_s}, S_{q_s}$. Denote these discrete logarithms by $r_1, s_1, \ldots, r_{q_s}, s_{q_s}$. From the execution of \C, we know $\mathsf{dt}(j) = (i,k,d,c,z) \ne \perp$ if and only if \C queries \textsc{Dlog} on $R_j S_j^d X_i^c$. Therefore, we denote the \textsc{Dlog} query on $R_j S_j^d X_i^c$ as the \textsc{Dlog} query associated with $\mathsf{dt}(j)$. For each $j \in [1..q_s]$, there are the following cases.

\medskip
\noindent \textbf{Case 0:} Both $\mathsf{dt}(j)$ and $\mathsf{dt}'(j)$ are $\perp$. In this case, $\mathcal{B}_1$ queries $r_j\leftarrow \textsc{Dlog}(R_j)$ and $ s_j\leftarrow \textsc{Dlog}(S_j)$.

\medskip
\noindent \textbf{Case 1:} Exactly one of $\mathsf{dt}(j)$ and $\mathsf{dt}'(j)$ is not $\perp$. Without loss of generality, assume $\mathsf{dt}(j) = (i,k,d,c,z)$, which implies $g^z = R_j S_j^d X_i^c$. $\mathcal{B}_1$ queries $s_j\leftarrow \textsc{Dlog}(S_j)$ and computes $r_j = z - d\cdot s_j - c\cdot f(\mathsf{id}_i)$.

\medskip
\noindent For all the following cases, both $\mathsf{dt}(j)$ and $\mathsf{dt}'(j)$ are not $\perp$ and we denote $(i,k,d,c,z) \leftarrow \mathsf{dt}(j)$ and $(i',k',d',c',z') \leftarrow \mathsf{dt}'(j)$.

\medskip
\noindent \textbf{Case 2:} $k \ne k'$ or $k = k' > k^*$. In this case, we have $d = h_k \ne h_{k'}' = d'$ and $g^z = R_j S_j^d X_i^c$, $g^{z'} = R_j S_j^{d'} X_i^{c'}$. Therefore, $\mathcal{B}_1$ computes
$s_j = \frac{z - c f(\mathsf{id}_i) - z' + c' f(\mathsf{id}_i)}{d - d'}$ and $r_j = z - d\cdot s_j - c\cdot f(\mathsf{id}_i)$.

\medskip
\noindent \textbf{Case 3:} $k = k' = k^*$. In this case, $\mathcal{B}_1$ computes $r_j, s_j$ the same as Case 1.

\medskip
\noindent \textbf{Case 4:} $k = k' < k^*$. $\mathcal{B}_1$ computes $r_j, s_j$ the same as Case 1. Also, in this case we have $d = d'$ and $c = c'$. Therefore, $\mathcal{B}_1$ queries the \textsc{Dlog} oracle once in order to simulate the \textsc{Dlog} queries associated with $\mathsf{dt}(j)$ and $\mathsf{dt}'(j)$.

We now count the number of \textsc{Dlog} queries made by $\mathcal{B}_1$. $\mathcal{B}_1$ queries the \textsc{Dlog} oracle $|CS|$ times to simulate the queries $\textsc{Dlog}(X_i)$ made by \C for each $i \in CS$. If $S_2(lr^*)<t-|CS|$ occurs, $\mathcal{B}_1$ queries the \textsc{Dlog} oracle $|D'|$ times to compute $a_0, \ldots, a_{t-1}$. For each $j \in [1..q_s]$, $\mathcal{B}_1$ queries \textsc{Dlog} twice to simulate the queries associated with $\mathsf{dt}(j)$ and $\mathsf{dt}'(j)$ and to compute $r_j, s_j$ in cases 0, 1, 2, 4, and three times in \textbf{Case} 3.
Since the condition of \textbf{Case 3} is equivalent to $j \in T_{\mathsf{dt}} \cap T_{\mathsf{dt}'}$, the total \textsc{Dlog} queries made by $\mathcal{B}_1$ is
$2q_s + |T_{\mathsf{dt}} \cap T_{\mathsf{dt}'}| + |CS| + |D'| = 2q_s + t - 1$.
Therefore, $\mathcal{B}_1$ breaks the OMDL assumption. \qed
\end{proof}

Lemma \ref{lem:tssuf4_2} below is devoted to the event $E_2$, which captures forgeries under the condition $S_2(lr) \neq S_3(lr)$. Since $\mathsf{st.mapPP}(pp_i) \neq \perp$ is enforced in the partial-signing algorithm $\mathsf{PS}$, an honest signer generates a $pp_i$ (and $S_3(lr)$ grows) whenever a $psig_i$ is requested (and $S_2(lr)$ grows), so forgeries with $S_3(lr) \subset S_2(lr)$ will not occur. Therefore, we only consider the case $S_2(lr) \subset S_3(lr)$.
In other words, $E_2$ captures the attack in which at least one $pp_i$ is honestly generated but the corresponding $psig_i$ is forged. Our proof idea is to embed OMDL challenges into the honest signer's $pp_i$, which would fall in $S_3(lr)\setminus S_2(lr)$ when $E_2$ occurs. 
We show that by using the Lemma~\ref{lem:mimfork} with an approach inspired by the mix forking technique \cite{CGRS23}, all secret exponents in the forger's token authentication transcripts $(u,V)$ (i.e., those belonging to $lr^*.\mathcal{S}\setminus S_3(lr^*)$) can be recovered, and an OMDL challenge in $S_3(lr^*)\setminus S_2(lr^*)$ can be solved without querying the \textsc{Dlog} oracle.

\begin{lemma}
\label{lem:tssuf4_2}
  There exists an OMDL adversary $\mathcal{B}_2$ issuing at most $2q_s$ queries to
  the \textsc{Chall} oracle and $q$ queries to random oracle such that
  \[
    \Pr[E_2] \;\le\;
    8 \mathsf{Adv}_{\mathbb{G}}^{\mathsf{omdl}}(\mathcal{B}_2) + \mathsf{negl}.
  \]
  Moreover, $\mathcal{B}_2$ runs in time
    $T(\mathcal{B}_2) \le \frac{16\mathsf{ns}^2q \ln(16\mathsf{ns}/\Pr[E_2])}{\Pr[E_2]}\cdot T(\F{})+T(q)$
  where $\mathsf{ns}$ is the number of signers, and $T(q)$ is the time required to perform exponentiations and group operations.
\end{lemma}

\begin{proof}
Let \F{} be an adversary as described in the theorem. Denote the output message-signature pair of \F{} as $(M^*, sig^* = (\tilde{R}^*, z^*))$. Without loss of generality, we assume \F{} always queries \textsc{RO} on $\mathsf{H_{sig}}(\tilde{X}, M^*, \tilde{R}^*)$ before \F{} returns and always queries \textsc{RO} on $\mathsf{H_{pp}}(\tilde{X}, X_i, R, S, V)$ and $\mathsf{H_{non}}(\tilde{X}, lr)$ prior to the query $\textsc{PSignO}(\mathsf{id}_i, lr)$ for some $\mathsf{id}_i$ and $lr$.
This adds up to $2q_s$ additional \textsc{RO} queries, and we let $q = q_h + 2q_s + 1$ where $q_h=q_{pp} + q_{non} + q_{sig}$.

Let $lr^*$ be the leader query whose nonce query $\mathsf{H_{non}}(\tilde{X}, lr^*)$ is the last such query before the signature hash query $\mathsf{H_{sig}}(\tilde{X}, M^*, \tilde{R}^*)$ and which satisfies $\mathsf{SVf}(\tilde{X},$ $ lr^*, \sigma^*) = \mathsf{true}$. Denote the event $E_2$ as
  \[
    E_2 :\mathsf{Vf}(vk, M^*, \sigma^*) ~\wedge~ lr^* \ne \perp ~\wedge~ S_2(lr^*)\subset S_3(lr^*).
  \]  
To begin, we construct an algorithm \C which is compatible with the syntax in Lemma \ref{lem:mimfork}.
The input of \C consists of $2q_s$ random elements $\mathbf{R} = \{R_1,\dots,R_{q_s}\}$, $\mathbf{S} = \{S_1,\dots,S_{q_s}\}$ from \G, and uniformly random integers $h_1,\dots,h_{q} \in \mathbb{Z}_p^*$. 
Also, \C has access to an oracle \textsc{Dlog}, which on input $X\in \G$ outputs $\mathsf{DL}_{\G,g}(X)$.
\C initializes the signers' states $\mathsf{st}_0,\dots,\mathsf{st}_{\mathsf{ns}}$,
counters $\mathsf{ctr}_h=0$ and $\mathsf{ctr}_s=0$. It then runs \F{}, answering its oracle queries as follows.

  \begin{itemize}
  \item $\widetilde{\textsc{Init}}(CS)$:
    \C picks $a_0,\dots,a_{t-1}\!\xleftarrow{\$}\mathbb{Z}_p$,
    sets $f(x)=\sum_{i=0}^{t-1}a_ix^i$, and defines
    $\tilde{X}=g^{f(0)}$, $X_i=g^{f(\mathsf{id}_i)}$ for $i\in[1..\mathsf{ns}]$.
    For $\mathsf{id}_i\in CS$ it stores $\mathsf{st}_i.\mathsf{sk}=f(\mathsf{id}_i)$.
    It returns $(\tilde{X}, (X_1,\dots,X_{\mathsf{ns}}),
    \{\mathsf{st}_i.\mathsf{sk}\}_{\mathsf{id}_i\in CS})$.
  
  \item $\widetilde{\textsc{RO}}[\mathsf{H}_{\mathsf{pp}}](x)$:
    If $\mathsf{H}_{\mathsf{pp}}(x)\neq \perp$, \C returns $\mathsf{H}_{\mathsf{pp}}(x)$. Otherwise,
    increments $\mathsf{ctr}_h$ to set $\mathsf{H}_{\mathsf{pp}}(x)\leftarrow h_{\mathsf{ctr}_h}$ and returns $\mathsf{H}_{\mathsf{pp}}(x)$.
  
  \item $\widetilde{\textsc{RO}}[\mathsf{H}_{\mathsf{non}}](x)$:
  If $\mathsf{H}_{\mathsf{non}}(x)\neq \perp$, \C returns $\mathsf{H}_{\mathsf{non}}(x)$. 
  Otherwise,
  increments $\mathsf{ctr}_h$ to set $\mathsf{H}_{\mathsf{non}}(x)\leftarrow h_{\mathsf{ctr}_h}$ and returns $\mathsf{H}_{\mathsf{non}}(x)$.
      
  \item $\widetilde{\textsc{RO}}[\mathsf{H}_{\mathsf{sig}}](x)$:
  If $\mathsf{H}_{\mathsf{sig}}(x)\neq \perp$, \C returns $\mathsf{H}_{\mathsf{sig}}(x)$. 
  Otherwise, 
  increments $\mathsf{ctr}_h$ to set $\mathsf{H}_{\mathsf{non}}(x)\leftarrow h_{\mathsf{ctr}_h}$ and returns $\mathsf{H}_{\mathsf{non}}(x)$.
  
  \item $\widetilde{\textsc{PPO}}(\mathsf{id}_i)$:
    Same as in the game $\mathbf{G}_{\mathsf{FROST2+}}^{\textsf{ts-suf-4}}$, except that, in the simulation of $\textsf{SPP}$, \C selects $u\!\xleftarrow{\$}\mathbb{Z}_p$, increments $\mathsf{ctr}_h$ and $\mathsf{ctr}_s$ to set $e=h_{\mathsf{ctr}_h}$,
    and computes $V =
      g^{u}/X_i^{e}R_{\mathsf{ctr}_s}^{e^2}
      S_{\mathsf{ctr}_s}^{e^3}$.    
    \C sets $\mathsf{H}_{\mathsf{pp}}(\tilde{X},X_i,R_{\mathsf{ctr}_s},S_{\mathsf{ctr}_s},
    V) \gets e$ and adds $pp_i \gets(R_{\mathsf{ctr}_s},S_{\mathsf{ctr}_s},u,V)$ to $\mathsf{st}_i.\mathsf{curPP}$.

  \item $\widetilde{\textsc{PSignO}}(\mathsf{id}_i, lr)$:
    Same as in the game $\mathbf{G}_{\mathsf{FROST2+}}^{\textsf{ts-suf-4}}$, except that, in the simulation of \textsf{PS}, if $\mathsf{st}_i.\mathsf{mapPP}(pp_{\mathsf{id}_i})\neq \perp$ and all $pp_j\in lr.\mathsf{PP}$ pass the verifications, \C retrieves $(R,S)\gets pp_{\mathsf{id}_i}$ and other $(R_j,S_j)\in lr.\mathsf{PP}$ to compute $\tilde{R}=\prod_{i\in lr.\mathcal{S}} R_i (\prod_{i\in lr.\mathcal{S}} S_i)^d$ where $d\gets\mathsf{H_{non}}(\tilde{X},lr)$.
    \C then removes $pp_{\mathsf{id}_i}$ from $\mathsf{st}_i.\mathsf{curPP}$ and returns
    $z_i \gets \textsc{Dlog}\bigl( R S^{d}\bigr) + c
      \Delta_{\mathsf{id}_i,lr.\mathcal{S}}(0)\,f(\mathsf{id}_i)$    
    where $c\gets\mathsf{H_{sig}}(\tilde{X},lr.\mathsf{msg},\tilde{R})$.
  \end{itemize}

  Finally, $\mathcal{F}$ outputs its forgery $(lr^*=\{M^*,\mathcal{S},\mathsf{PP}\},\sigma^*=(\tilde{R}^*,z^*))$ and wins the game $\mathbf{G}_{\mathsf{FROST2}\texttt{+}}^{\textsf{ts-suf-4}}$. 
  \C returns $(I,\{\mathsf{out}_i\}_{i\in I},\theta)$ where $I$ are the set of random oracle indices such that $\mathsf{out}_i$ contains $h_i\gets\mathsf{H}_{\mathsf{pp}}(\tilde{X},X_i,R_i,S_i,V_i)$ for $\mathsf{id}_i\in lr^*.\mathcal{S}\setminus S_3(lr^*)$ and $\theta=(lr^*,\sigma^*)$.

\paragraph{Analysis of \C.}
There are two possible forging cases under the condition $S_2(lr^*)\subset S_3(lr^*)$: (1) $lr^*.\mathcal{S}\setminus S_3(lr^*) = \emptyset$; and (2) $lr^*.\mathcal{S}\setminus S_3(lr^*) \neq \emptyset$.
We first present the analysis for case (1) and subsequently explain how to extend it to handle case (2).\\


\noindent
\textbf{Case (1):}
When $lr^*.\mathcal{S}\setminus S_3(lr^*) = \emptyset$, $\F{}$ forges without generating any pre-processing tokens such that $S_3(lr^*) = S_4(lr^*) = lr^*.\mathcal{S}$.
We show that there is an OMDL adversary $\mathcal{B}_2$ which uses \C as the helper algorithm to compute the discrete logarithms of $(R^*,S^*)$ for an $\mathsf{id}^*\in S_3(lr^*)\setminus S_2(lr^*)$.
$\mathcal{B}_2$ queries $2q_s$ challenges $\mathbf{R} = \{R_1,\dots,R_{q_s}\}$, $\mathbf{S} = \{S_1,\dots,S_{q_s}\}$ from \textsc{Chall} and selects random $h_1,\ldots,h_q\in\Zs{p}$ as the input to \C.
The \textsc{Dlog} queries from \C are relayed by $\mathcal{B}_2$ to its \textsc{Dlog} of $\mathbf{G}_{\G}^{\mathsf{omdl}}$.
From $(I,\{\mathsf{out}_i\}_{i\in I},\theta)\gets\C$, $\mathcal{B}_2$ needs only $\theta$ to obtain $((z^*,c^*,f(0)),\tilde{R}^*=\prod_{\mathsf{id}_j\in lr^*.\mathcal{S}} R_j(\prod_{\mathsf{id}_j\in lr^*.\mathcal{S}} S_j)^{d^*})$ and compute the sum
$\alpha=\sum_{\mathsf{id}_j\in lr^*.\mathcal{S}} r_j+d^*\cdot s_j = z^*-c^*f(0)$.

Subsequently, for all $(R_j,S_j)$ in the set $D= S_3(lr^*)\setminus S_2(lr^*)$,
$\mathcal{B}_2$ queries $\{r_j\gets\textsc{Dlog}(R_j)\}_{\mathsf{id}_j\in D}$ and 
$\{s_j\gets\textsc{Dlog}(S_j)\}_{\mathsf{id}_j\in D\setminus \{\mathsf{id}^*\}}$ where $\mathsf{id}^*\in D$ is the targeted identity who owns $(R^*,S^*)$.
For those $(R_j,S_j)$ belonging to $\mathsf{id}_j\in S_2(lr^*)$, they have been queried to \textsc{PSignO} and $\mathcal{B}_2$ knows
$z_j= r_j + d^*\cdot s_j+c^*\Delta_{\mathsf{id}_j,lr^*.\mathcal{S}}(0)f(\mathsf{id}_j)$
in which $\mathcal{B}_2$ can query $s_j\leftarrow \textsc{Dlog}(S_j)$ and computes $r_j=z_j-c^*\Delta_{\mathsf{id}_j,lr^*.\mathcal{S}}(0)f(\mathsf{id}_j)-d^* \cdot s_j$.
Therefore, $lr^*.\mathcal{S}=D\cup S_2(lr^*)$ and $\mathcal{B}_2$ can extract $s^*\gets \mathsf{DL}_{\G,g}(S^*)$ by computing:
\begin{align*}
    s^*=\frac{\alpha - \sum_{\mathsf{id}_j\in lr^*.\mathcal{S}} r_j - d^*\sum_{\mathsf{id}_j\in lr^*.\mathcal{S}\setminus\{\mathsf{id}^*\}} s_j}{d^*}.
\end{align*}

For the OMDL challenges not involved in the forgery $(R_j,S_j)\not\in lr^*.\mathsf{PP}(\mathsf{id}_j)$ where $\mathsf{id}_j\in HS$, they are either queried or have not queried to \textsc{PSignO}.
If queried, $\mathcal{B}_2$ knows the corresponding partial signature $z_j$ in which it can query $s_j\leftarrow \textsc{Dlog}(S_j)$ and compute $r_j=z_j-f(\mathsf{id}_j)c_j\Delta_{\mathsf{id}_j,lr.\mathcal{S}}(0)-d_j s_j$.
Else, $\mathcal{B}_2$ queries $r_j\leftarrow \textsc{Dlog}(R_j)$ and $ s_j\leftarrow \textsc{Dlog}(S_j)$.

We now count the number of \textsc{Dlog} queries made by $\mathcal{B}_2$. 
For each challenge pair in $\{R_j=g^{r_j},S_j=g^{s_j}\}_{j=1}^{q_s}$, $\mathcal{B}_2$ queries \textsc{Dlog} twice except the $s^*\leftarrow\mathsf{DL}_{\G,g}(S^*)$ used by the targeted honest signer $\mathsf{id}^*\in S_3(lr^*)\setminus S_2(lr^*)$. 
Hence the total $\textsc{Dlog}$ queries made by $\mathcal{B}_2$ equals $2q_s - 1$, breaking the OMDL assumption. 
Denote \textsf{BadHash} as the event of $h_i=h_j$ for some $i\ne j \le q$. Since every hash value is selected independently, it is $\Pr[\mathsf{BadHash}]=1/p$. 
Absorbing all strictly negligible terms into a function $\mathsf{negl}$, we have:
  \[
    \mathsf{Adv}_{\mathbb{G}}^{\mathsf{omdl}}(\mathcal{B}_2) \ge \Pr[\mathbf{G}_{\mathsf{FROST2+}}^{\textsf{ts-suf-4}}(\F{})=\mathsf{true}] - \Pr[\mathsf{BadHash}] \ge \Pr[E_2] - \mathsf{negl},
  \]
and $T(\mathcal{B}_2) \approx T(\F{}) + T(q)$.\\

\noindent
\textbf{Case (2):}
Now, we consider the case where $\F{}$ forges with at least one pre-processing token in $lr^*.\mathsf{PP}$ generated by itself, i.e., $lr^*.\mathcal{S}\setminus S_3(lr^*) \neq \emptyset$ where $S_3(lr^*)\subset S_4(lr^*)\subseteq lr^*.\mathcal{S}$. 
In this case, the output of \C's single run is not sufficient for solving the OMDL problem. 
This is because \C does not know the discrete logarithms of $\{R_i,S_i\}_{\mathsf{id}_i\in lr^*.\mathcal{S}\setminus S_3(lr^*)}$ generated by \F{}.
To overcome this obstacle, we need a helper algorithm $\C'$. 

\paragraph{Construct $\C'$ from $\mathsf{MIMFork}_{N,2}^{\C}$.}
We now construct the helper algorithm $\C'$, which is a syntactically valid OMDL adversary, using $\mathsf{MIMFork}_{N,2}^{\C}$ with $N=|I|=|lr^*.\mathcal{S}\setminus S_3(lr^*)|$.
$\C'$ initiates
\C the same as $\mathcal{B}_2$ did in the previous case.
Similar to the proof in case (1), from $(\mathcal{O},\mathcal{O}')\gets\mathsf{MIMFork}_{N,2}^{\C}$, $\C'$ also computes the values $(\alpha, \{r_j\gets\textsc{Dlog}(R_j)\}_{\mathsf{id_j}\in S_3(lr^*)}, \{s_j\gets\textsc{Dlog}(S_j)\}_{\mathsf{id_j}\in S_3(lr^*)\setminus\{\mathsf{id}^*\}})$.
Now, $\C'$ needs to additionally handle the transcripts $\{(e_{i,j}=h_{i,j},v_i+r_ie_{i,j}^2+s_ie_{i,j}^3=u_{i,j}-f(\mathsf{id}_i)e_{i,j},V_i)\gets pp_{i,j}\}_{j=0}^2$ for $\mathsf{id}_i\in D$ where $D=lr^*.\mathcal{S}\setminus S_3(lr^*)$.
These yields a solvable system of linear equations, allowing $\mathcal{C}'$ to extract the discrete logarithms $\{r_i\gets\mathsf{DL}_{\G,g}(R_i), s_i\gets\mathsf{DL}_{\G,g}(S_i)\}_{\mathsf{id}_i\in D}$.
As $lr^*.\mathcal{S}=D\cup S_3(lr^*)$, $\C'$ can compute 
$$\mathsf{DL}_{\G,g}(S^*)\rightarrow s^*=\frac{\alpha - \sum_{\mathsf{id}_i\in lr^*.\mathcal{S}} r_i - d^*\sum_{\mathsf{id}_i\in lr^*.\mathcal{S}\setminus\{\mathsf{id}^*\}} s_i}{d^*}.$$
$\C'$ then solves those OMDL challenges not involved in the forgery by querying the \textsc{Dlog} oracle.
Since a hash collision occurs with a negligible probability only, by Lemma \ref{lem:mimfork}, we have:
  \[
    \mathsf{acc}(\C')\ge\mathsf{acc}(\mathsf{MIMFork}_{N,2}^{\C}) - \mathsf{negl} \ge \hat{\varepsilon}/8 - \mathsf{negl} \ge \Pr[E_2]/8 - \mathsf{negl}.
  \]
However, after broadcasting a $pp_i$ in any of the $|I|\cdot T_{max}$ rewinds in $\mathsf{MIMFork}^{\C}_{N,2}$, \F{} still can issue \textsc{PSignO} queries till it forges a signature. 
So, $\C'$ does not solve OMDL because the $|I|\cdot T_{max}$ rewinds make \C issues more than $2q_s$ \textsc{Dlog} queries to answer these additional \textsc{PSignO} queries.

\paragraph{Construct $\mathcal{B}_2$ from $\C'$.}
We now construct an OMDL adversary $\mathcal{B}_2$ which is defined the same as $\C'$ except that, starting from the first rewind of each instance $i\in I$, $\mathcal{B}_2$ aborts any executions of $\C$ after it has obtained the $pp_i$. 
This avoids $\mathcal{B}_2$ from issuing too many \textsc{Dlog} queries as done by $\C'$.
Moreover, $\mathcal{B}_2$ has the same success probability as $\C'$ because the $\theta$ required to compute $\alpha$ is available after the first run ended completely, that is, before a rewind begins.
Since $N=|I|=|lr^*.\mathcal{S}\setminus S_3(lr^*)| \le \mathsf{ns}-1$, and every $i\in I$ collected $n=2$ forks, $\mathcal{B}_2$ uses time:
\begin{align*}
    T(\mathcal{B}_2)\le T(\C') &\le |I|\cdot T_{max}\cdot T(\C)+T(q)
    \le |I|\cdot T_{max}\cdot T(\F{})+T(q)\\
    &\le \mathsf{ns}\cdot \frac{8\cdot \mathsf{ns}\cdot 2q \ln(8\cdot\mathsf{ns}\cdot 2/\Pr[E_2])}{\Pr[E_2]}\cdot T(\F{})+T(q)\\
    &\le \frac{16\mathsf{ns}^2q \ln(16\mathsf{ns}/\Pr[E_2])}{\Pr[E_2]}\cdot T(\F{})+T(q).
\end{align*}
This completes the proof.   \qed
\end{proof}

Up to this point, \textsf{FROST2}\texttt{+} has achieved the \textsf{TS-SUF-3} security whose forgeries are covered by the events $E_1$ and $E_2$. 
If we apply the \textsf{ATS[TS,DS]} transformation framework as in \textsf{FROST} \cite{BTZ22,BCKMTZ22}, we obtain an \textsf{ATS[FROST2+, DS]} scheme that achieves \textsf{TS-SUF-4} security. 
However, the transformation comes with the cost of managing an extra DS key pair, longer $pp_i$ and slower signing time.
To avoid these, we give a direct proof instead, which is the proof for event $E_3$.
As $S_4(lr)=HS\cap lr^*.\mathcal{S}$ and the event $E_2$ has covered $S_2(lr)\neq S_3(lr)$, we only need to deal with the forgeries under $S_2(lr)=S_3(lr)\subset S_4(lr)$ in $E_3$.

Before presenting the direct proof, we would like to introduce a leader request forger \F{lr} of \textsf{FROST2}\texttt{+} which possesses the same ability as that of a normal \textsf{TS-SUF-4} forger \F{}, but instead of producing a valid signature $(\sigma^*,lr^*)$ as the output, \F{lr} produces a valid (i.e. $S_3(lr^*)\subset S_4(lr^*)$) leader request $lr^*$ as the output.
Denote $A \Longrightarrow B$ as $A$ implies $B$, that is, $B$ is reducible to $A$.
Below, we show that $\F{lr}\Longleftrightarrow\F{}$ when a forgery falls under $S_3(lr)\subset S_4(lr)$, that is, when at least a pre-processing token is forged.

\begin{theorem}[$\F{lr}\Longleftrightarrow\F{}$]
\label{thm:F_lr<->F}
For \textsf{FROST2}\texttt{+}, a \textsf{TS-SUF-4} forger \F{} which output a forgery that falls under the event 
\[
    E :\; \mathsf{Vf}(\tilde{X}, M^*, \sigma^*) ~\wedge~ lr^* \ne \perp ~\wedge ~ S_3(lr^*)\subset S_4(lr^*)
\]
is equivalent to a leader request forger \F{lr}. 
\end{theorem}
\begin{proof}
The proof follows from Lemmas \ref{lem:Flr->F} and \ref{lem:F->Flr} which show that \F{} and \F{lr} are equivalent under $E$ such that they are reducible to each other.  \qed
\end{proof}

\begin{lemma}[$\F{}\Longrightarrow\F{lr}$]
\label{lem:Flr->F}
Let \F{} be a \textsf{TS-SUF-4} forger for \textsf{FROST2}\texttt{+}, and \F{lr} be a leader request forger in \textsf{FROST2}\texttt{+}.
There exists an \F{lr} which wins $\mathbf{G}_{\textsf{FROST2}\texttt{+}}^{\textsf{ts-lr}}$ such that 
\[
\mathsf{Adv}_{\textsf{FROST2}\texttt{+}}^{\mathsf{ts\text{-}lr}}(\F{lr})\ge \Pr[E] - \mathsf{negl},
\]
where $E$ is the event of Theorem~\ref{thm:F_lr<->F}. Moreover, the running time of \F{lr} is at most $T(\F{}) + O(q)$ where $O(q)$ is the wrapper overhead that relays oracle queries.
\end{lemma}
\begin{proof}
It is straight forward to see that \F{lr} can simulate the environment perfectly for \F{} by acting as the middle-man in between its underlying  \textsf{FROST2}\texttt{+} oracles and \F{}.
When \F{} outputs a valid \textsf{TS-SUF-4} forgery $(\sigma^*,lr^*)$, \F{lr} uses $lr^*$ as its output. 
In fact, \F{lr} can return $lr^*$ before answering the \textsc{PSignO} query on $lr^*$, that is, whenever it sees that a valid $lr^*$ contains $pp_i$ such that $S_3(lr^*)\subset S_4(lr^*)$.
Therefore, \F{lr} runs in time at most that of \F{}.
\qed
\end{proof}

\begin{lemma}[$\F{lr}\Longrightarrow\F{}$]
\label{lem:F->Flr}
Let \F{} be a \textsf{TS-SUF-4} forger for \textsf{FROST2}\texttt{+}, and \F{lr} be a leader request forger in \textsf{FROST2}\texttt{+}.
There exists an \F{} which wins $\mathbf{G}_{\textsf{FROST2}\texttt{+}}^{\textsf{ts-suf-4}}$ such that 
\[
\mathsf{Adv}_{\textsf{FROST2}\texttt{+}}^{\mathsf{ts\text{-}lr}}(\F{lr})\le\Pr[E]\le 8t\cdot\mathsf{Adv}_{\textsf{FROST2}\texttt{+}}^{\mathsf{ts\text{-}suf\text{-}4}}(\F{})
\]
where $E$ is the event of Theorem~\ref{thm:F_lr<->F} and $t\le \mathsf{ns}$.
Moreover, \F{} runs in time $T(\F{}) \le \frac{24tq \ln(24t/\mathsf{Adv}_{\textsf{FROST2}\texttt{+}}^{\mathsf{ts\text{-}lr}}(\F{lr}))}{\mathsf{Adv}_{\textsf{FROST2}\texttt{+}}^{\mathsf{ts\text{-}lr}}(\F{lr})}\cdot T(\F{lr})+T(q)$ where $q$ is the number of hash queries issued and $T(q)$ is the time required to perform exponentiations and group operations.
The number of \F{}’s queries to each oracle
is at most that of \F{lr}.
\end{lemma}
\begin{proof}
We show that if there exists a leader request forger \F{lr}, then we can construct a \textsf{TS-SUF-4} forger \F{} to win in $\mathbf{G}_{\textsf{FROST2}\texttt{+}}^{\textsf{ts-suf-4}}$ such that the output forgery $(lr^*,\sigma^*)$ falls under the event:
\[
    E :\; \mathsf{Vf}(\tilde{X}, M^*, \sigma^*) ~\wedge~ lr^* \ne \perp ~\wedge ~ S_3(lr)\subset S_4(lr).
\]

\F{} constructs a wrapper \C whose task is only to relay all oracle queries from \F{lr} to \F{}, which then replies by querying its own oracles.
To begin, \F{} forwards the $\widetilde{\textsc{Init}}(CS)$ query from \C to its underlying \textsc{Init} oracle and returns the required parameter $(\tilde{X},\textsf{aux},\{x_i\}_{\mathsf{id}_i\in CS})$.
Without loss of generality, we assume $|CS|=t-1$. 
If $|CS|<t-1$, \F{} removes a random $\widetilde{CS}$ of size $t-1-|CS|$ from $HS$ and queries $\textsc{Init}(CS\cup \widetilde{CS})$ but does not return the signing key share of $\mathsf{id}_j\in\widetilde{CS}$ to \C.

When \F{lr} returns a valid $lr^*=(M,\mathcal{S},\{pp_i=(R_i,S_i,u_i,V_i)\}_{\mathsf{id}_i\in\mathcal{S}})$ during the first run, \C knows $(I,\mathsf{out}_{i^*},\theta=\perp)$ with $I=\{i^*\}=S_4(lr^*)\setminus S_3(lr^*)$ corresponds to a query $e_{i^*}=\mathsf{H_{pp}}(\tilde{X},X_{i^*},R_{i^*},S_{i^*},V_{i^*})$ such that $pp_{i^*}\in lr^*$.
If $1<|I|\le \mathsf{ns}$, \C selects a random index $i^*\in I$. 
Since $\widetilde{CS}\in HS$ is chosen randomly by \C and is hidden from the view of \F{lr}, the probability of \C getting an invalid $i^*\in I$ is at most
$\Pr[i^*\in\widetilde{CS}]= \frac{|\widetilde{CS}|}{|HS|}\le \frac{t-1}{t}$. 
Then \C returns $(e_{i^*},\mathsf{out}_{i^*})$.

\F{} executed $(\mathcal{O},\mathcal{O}')\gets \mathsf{MIMFork}^{\C}_{1,3}$ to get four transcripts $$(\tilde{X},X_{i^*},R_{i^*},S_{i^*},V_{i^*},u_{i^*},e_{i^*})\gets\mathcal{O},(\{u_{{i^*},j},e_{{i^*},j}\}_{j=1}^3)\gets\mathcal{O}'$$ 
which yields a solvable system of linear equations and allows \F{} to extract $\mathsf{st}_{i^*}.\mathsf{sk}=x_{i^*}\gets\mathsf{DL}_{\G,g}(X_{i^*}),r_{i^*}\gets\mathsf{DL}_{\G,g}(R_i),s_{i^*}\gets\mathsf{DL}_{\G,g}(S_{i^*})$.
Thus, \F{} knows $|CS|+1=t$ pairs of $(\mathsf{id}_i,\mathsf{st}_i.\mathsf{sk}=x_i=f(\mathsf{id}_i))$ and can extract the shared signing key $f(0)$ through Lagrange interpolation.
\F{} then utilizes $f(0)$ to generate a signature $\sigma^*$ on $lr^*$ and output $(\sigma^*, lr^*)$ to win its $\mathbf{G}_{\textsf{FROST2}\texttt{+}}^{\textsf{ts-suf-4}}$ game.

Throughout the game, all $q$ hash values are produced by \F{}'s hash oracles and we assume hash collision does not occur, or \F{} can win without \F{lr}.
Since \F{} simulates the environment perfectly, by Lemma \ref{lem:mimfork}, we have:
\[
\mathsf{Adv}_{\textsf{FROST2}\texttt{+}}^{\mathsf{ts\text{-}suf\text{-}4}}(\F{})=\mathsf{acc}(\mathsf{MIMFork}^{\C}_{1,3}) \ge (1-\Pr[i^*\in\widetilde{CS}])\frac{\Pr[E]}{8}=  \frac{\mathsf{Adv}_{\textsf{FROST2}\texttt{+}}^{\mathsf{ts\text{-}lr}}(\F{lr})}{8t}
\]
and \F{} runs in time $T(\F{}) \le \frac{24tq \ln(24t/\mathsf{Adv}_{\textsf{FROST2}\texttt{+}}^{\mathsf{ts\text{-}lr}}(\F{lr}))}{\mathsf{Adv}_{\textsf{FROST2}\texttt{+}}^{\mathsf{ts\text{-}lr}}(\F{lr})}\cdot T(\F{lr})+T(q)$.
\qed

\end{proof}

Now we are ready to prove security under event $E_3$ using an approach similar to that of Lemma \ref{lem:tssuf4_1}, by constructing an OMDL challenger $B_3$ that uses the forger \F{lr}, which is equivalent to the \textsf{TS-SUF-4} forger $\F{}$.
A careful reader will notice that we do not rely on \F{}.
This is deliberate and required for proof soundness, not for reduction tightness. 
In particular, if we used \F{}, we would exhaust the \textsc{Dlog} quota when answering \textsc{PSignO} queries, preventing us from performing the computation needed (in Case 3 of Lemma \ref{lem:tssuf4_1}) to extract the discrete logarithms of $(R,S)$ used to construct $lr^*$. 

\begin{lemma}
\label{lem:tssuf4_3}
  There exists an OMDL adversary $\mathcal{B}_3$ issuing at most $2q_s+t$ queries to
  the \textsc{Chall} oracle and $q$ queries to random oracle such
  \[    \Pr[E_3] \le 8\mathsf{Adv}_{\mathbb{G}}^{\mathsf{omdl}}(\mathcal{B}_3) + \mathsf{negl}.
  \]
Moreover, $\mathcal{B}_3$ runs in time $T(\mathcal{B}_3)\le \frac{24 q \ln(24/\Pr[E_3])}{\Pr[E_3]}\cdot T(\F{}) +T(q)$ where $\mathsf{ns}$ is the number of signers, and $T(q)$ is the time required to perform exponentiations and group operations.
\end{lemma}

\begin{proof}
The proof setup is similar to that of Lemma \ref{lem:tssuf4_1} but considering the event 
\[
    E_3 :\; \mathsf{Vf}(\tilde{X}, M^*, \sigma^*) ~\wedge~ lr^* \ne \perp ~\wedge ~ \neg(S_2(lr)=S_3(lr)= S_4(lr)).
  \]
As event $E_1$ has covered $|S_2(lr^*)|<t$ and $E_2$ has covered $S_2(lr^*)\neq S_3(lr^*)$, we only need to deal with the forgeries under $S_2(lr^*)=S_3(lr^*)\subset S_4(lr^*)$ in event $E_3$ such that $|S_2(lr^*)|=|S_3(lr^*)|\ge t$.
Furthermore, there must be at least $t+1$ honest signers in $lr^*.\mathcal{S}$ where $t$ fall under $S_2(lr^*)=S_3(lr^*)$ and one falls under $S_4(lr^*)\setminus S_3(lr^*)$.
Notice that $E_3\subset E$ for the event $E$ in Theorem \ref{thm:F_lr<->F}.
By Lemma \ref{lem:Flr->F}, from any \textsf{TS-SUF-4} forger $\F{}$ that wins under $E$ we can construct a leader request forger $\F{lr}$ with
\begin{equation*}
  \mathsf{Adv}_{\mathsf{FROST2+}}^{\mathsf{ts\text{-}lr}}(\F{lr}) \;\ge\; \Pr[E_3] - \mathsf{negl},
  \qquad
  T(\F{lr}) \;\le\; T(\F{}) + O(q),
\end{equation*}
where the $O(q)$ overhead is that of the wrapper that relays the oracle queries of $\F{}$ and outputs $lr^*$ as soon as it sees $S_3(lr^*)\subset S_4(lr^*)$. 

We simulate the key generation and oracles the same as that of Lemma \ref{lem:tssuf4_1} but here, \C runs \F{lr} instead of \F{}, and takes in $h_1,\ldots,h_q$ only as the hash values. 
Unlike \F{} which issues a \textsc{PSignO} query for $lr^*$ to generate partial signatures before forging,
\F{lr} outputs $lr^*$ directly without issuing this specific query in its winning condition.
Running $\F{lr}$ inside $\C$ (rather than $\F{}$) is what avoids exhausting the \textsc{Dlog} quota on the \textsc{PSignO} query for $lr^*$, exactly as discussed before the lemma.

After receiving the output $(lr^*=\{M^*,\mathcal{S},\mathsf{PP}=\{pp_i=(R_i,S_i,u_i,V_i)\}_{\mathsf{id}_i\in\mathcal{S}}\})$ from \F{lr}, \C returns $\perp$ if $\textsf{BadPPO}=\textsf{true}$, or $E_3$ does not occur. 
Else, denote $I_{\C}=S_4(lr^*)\setminus S_3(lr^*)$,
\C selects an index $i^*\in I_{\C}$ such that
$i^*$ corresponds to $e_{i^*}=\mathsf{H_{pp}}(\tilde{X}, X_{i^*}, R_{i^*}, S_{i^*}, V_{i^*})$ with $pp_{i^*}\in lr^*$.
Then, \C returns $(e_{i^*}, \mathsf{out}_{i^*})\gets\mathcal{O}$.
Let $\hat\varepsilon := \mathsf{Adv}_{\mathsf{FROST2+}}^{\mathsf{ts\text{-}lr}}(\F{lr})$ denote the success probability of $\C$.

\paragraph{Construct $\mathcal{C}'$ from $\mathsf{MIMFork}^{\C}_{1,3}$.}
We can construct a helper algorithm $\mathcal{C}'$ which is a syntactically valid OMDL adversary using $\mathsf{MIMFork}^{\C}_{1,3}$ and the available \textsc{Dlog} oracle. 
To begin, $\mathcal{C}'$ queries its \textsc{Init} and then queries \textsc{Chall} oracle $2q_s + t$ times to generate the three lists of group elements $(\mathbf{A},\mathbf{R},\mathbf{S})$ as the input to run $\mathsf{MIMFork}^{\C}_{1,3}$. 
All \textsc{Dlog} queries from $\mathsf{MIMFork}^{\C}_{1,3}$ are forwarded by $\mathcal{C}'$ to the \textsc{Dlog} oracle in the game $\mathbf{G}_{\G}^{\mathsf{omdl}}$. 

We first show how to compute the discrete logarithm of $A_0, \ldots, A_{t-1}$.
From $(\mathcal{O}, \mathcal{O}')\gets\mathsf{MIMFork}^{\C}_{1,3}$, $\mathcal{C}'$ obtains $\{e_{i^*},pp_{i^*}=(R_{i^*},S_{i^*},u_{i^*},V_{i^*})\gets\mathsf{out}_{i^*}\}\gets\mathcal{O}$ and $\{e_{i^*,j},pp_{i^*}=(R_{i^*},S_{i^*},u_{i^*,j},V_{i^*})\}_{1\le j\le 3}\gets\mathcal{O}'$. 
This yields a solvable system of linear equations, allowing $\mathcal{C}'$ to extract the discrete logarithms $f(\mathsf{id}_{i^*})=x_{i^*}\gets\mathsf{DL}(X_{i^*}), r_{i^*}\gets\mathsf{DL}(R_{i^*})$ and $s_{i^*}\gets\mathsf{DL}(S_{i^*})$.
Without loss of generality, we assume $|CS|=t-1$ and so $\mathcal{C}'$ knows $|CS|+1 = t$ pairs of $(\mathsf{id}_i,f(\mathsf{id}_i))$.
If $|CS|<t-1$, $\mathcal{C}'$ can compute $f(\mathsf{id}_i)=x_i\gets\textsc{Dlog}(X_i)$ for a random set $\mathsf{id}_i\in D\subseteq (HS\setminus\{\mathsf{id}_{i^*}\})$ of size $|D|=t-1-|CS|$.
$\mathcal{B}_3$ can interpolate the key generation polynomial $f(\cdot)$ to extract all $a_0\gets\mathsf{DL}(A_0),\ldots,a_{t-1}\gets\mathsf{DL}(A_{t-1})$.
Since $lr^*$ have not been queried to \textsc{PSignO}, for the OMDL challenges used as tokens $(R_i,S_i)\in S_3(lr^*)$, $\mathcal{C}'$ simply queries the \textsc{Dlog} oracle to solve them.
For those $(R_i,S_i)\not\in S_3(lr^*)$ generated through the \textsc{PPO} queries before \F{lr} returns $lr^*$, they are either not queried, or queried once to \textsc{PSignO} or \textsc{Dlog}.
$\mathcal{C}'$ queries \textsc{Dlog} oracle to solve them as explained in the cases 0, 1, 2 and 4 in Lemma \ref{lem:tssuf4_1}.

However, $\mathcal{C}'$ does not break the OMDL assumption because during the $T_{max}$ rewinds, \F{lr}'s view changes after recomputing $u_{i^*,j}$ and it may continue issuing \textsc{PSignO} queries which then exceeds the \textsc{Dlog} quota.

\paragraph{Construct $\mathcal{B}_3$ from $\mathcal{C}'$.}
We now construct an OMDL adversary using $\mathcal{B}_3$ which is defined the same as $\mathcal{C}'$ except that, starting from the first rewind of the instance $i^*\in I$, $\mathcal{B}_3$ aborts any execution of \C after it has obtained the $\{e_{i^*,j},pp_{i^*}=(R_{i^*},S_{i^*},u_{i^*,j},V_{i^*})\}_{1\le j\le 3}\gets\mathcal{O}'$.
This avoids $\mathcal{B}_3$ from issuing too many \textsc{Dlog} queries as done by $\mathcal{C}'$.
Moreover, $\mathcal{B}_3$ has the same success probability as $\mathcal{C}'$ such that $\mathsf{Adv}_{\mathbb{G}}^{\mathsf{omdl}}(\mathcal{B}_3)=\mathsf{acc}(\mathcal{C}')$.
Therefore, $\mathcal{B}_3$ extracted all $2q_s+t$ discrete logarithms in the OMDL challenges by using only $2q_s+t-1$ \textsc{Dlog} queries.
Absorbing all strictly negligible terms such as $\Pr[\mathsf{BadPPO}]$ and hash collision into a function $\mathsf{negl}$, we have
\begin{align*}
\mathsf{Adv}_{\mathbb{G}}^{\mathsf{omdl}}(\mathcal{B}_3)=\mathsf{acc}(\mathcal{C}') &\ge \mathsf{acc}(\mathsf{MIMFork}^{\C}_{1,3}) - \mathsf{negl} \ge \hat{\varepsilon}/8 - \mathsf{negl} \qquad \text{//by Lemma \ref{lem:mimfork}}\\
&\ge \frac{\mathsf{Adv}_{\mathsf{FROST2+}}^{\mathsf{ts\text{-}lr}}(\F{lr})}{8} - \mathsf{negl} \\
&\ge \frac{\Pr[E_3]}{8} - \mathsf{negl} \qquad \text{//by Lemma \ref{lem:Flr->F}}
\end{align*}
and $\mathcal{B}_3$ runs in time:
\begin{align*}
T(\mathcal{B}_3) \le T(\mathcal{C}')&\le |I|\cdot T_{max}\cdot T(\F{lr}) +T(q)\\
&\le N\cdot\frac{ 8Nnq \ln(8Nn/\hat{\varepsilon})}{\hat{\varepsilon}}\cdot T(\F{lr}) +T(q)\\
&\le \frac{24 q \ln(24/\Pr[E_3])}{\Pr[E_3]}\cdot T(\F{}) +T(q),
\end{align*}
where the last line uses $\hat\varepsilon\ge \Pr[E_3]-\mathsf{negl}$ and $T(\F{lr})\le T(\F{})+O(q)$ with the resulting $O(q)\cdot \frac{24 q \ln(24/\Pr[E_3])}{\Pr[E_3]}$ overhead absorbed into $T(q)$.
 \qed
\end{proof}

\subsection{Security with Distributed Key Generation}
\label{sec:dkg_proof}
Our amendments on the pre-processing token generation and verification mechanism are orthogonal to key generation, allowing seamless integration with \textsf{FROST2} when instantiated with secure distributed key generation (DKG) protocols such as \textsf{PedPoP}~\cite{FROST,BCKMTZ22}. 
This follows because \textsf{SPP/PS} only use the existing key share $\mathsf{st}_i.\mathsf{sk}$ and token exponents $(r,s)$, and do not involve the DKG commitment/PoP messages.
To ensure completeness, we demonstrate that our proposed scheme maintain \textsf{TS-SUF-4} security in the ROM when it is instantiated in the DKG setting using \textsf{PedPoP} (in Appendix \ref{apdx:pedpop}).


Essentially, the oracle simulations are similar to that in Theorem \ref{thm:frost2+}, except for the key generation.
We first show how to build an OMDL adversary $\mathcal{B}_1'$ for a forgery under event $E_1$ in the DKG setting.
Intuitively, we construct the OMDL adversary $\mathcal{B}_1'$ using the mix forking technique \cite{CGRS23} as follows.
First, we define a wrapper \C that is same as that in Lemma \ref{lem:tssuf4_1} but its $\widetilde{\textsc{Init}}$ simulates the \textsf{PedPoP} protocol. 
Then, an algorithm $\mathcal{D}$ runs $\mathsf{Fork}^{\C}$ to extract the signing key $f(0)$.
Next, we define $\mathcal{D}'$ which is syntactically a valid OMDL adversary but does not break the OMDL assumption.
$\mathcal{D}'$ runs $\mathsf{MIMFork}^{\mathcal{D}}_{|CS|,1}$ and solves all the OMDL challenges.
Finally, we present $\mathcal{B}_1'$, which uses $\mathcal{D}'$ as a helper to break the OMDL assumption.

We do not discuss the DKG-counterpart $\mathcal{B}_2'$ for the $\mathcal{B}_2$ in Lemma \ref{lem:tssuf4_2} which knows the signing key shares of every honest signer and it can run the \textsf{PedPoP} protocol honestly.
Since $\mathcal{B}_1$ in Lemma \ref{lem:tssuf4_1} and $\mathcal{B}_3$ in Lemma \ref{lem:tssuf4_3} share the same game simulation strategy, 
the DKG-counterpart $\mathcal{B}_3'$ also adopts the same key generation simulation as $\mathcal{B}_1'$'s.
However, $\mathcal{B}_3'$ does not need to make use of the mix forking technique.
Specifically, $\mathcal{B}_3'$ runs $\mathsf{MIMFork}^{\C}_{1,3}$ as in Lemma \ref{lem:tssuf4_3} to extract $f(\mathsf{id}_{i^*})$ and subsequently the $f'(\mathsf{id}_{i^*})$.
This allows $\mathcal{B}_3'$ to perform Lagrange interpolation and solve an extra OMDL challenges without querying \textsc{Dlog}.

\begin{theorem}
\label{thm:frost2+dkg}
Let $\mathcal{F}$ be a \textsf{TS-SUF-4} forger against $\mathsf{FROST2+}$ with \textsf{PedPoP} making at most $q$ queries to the random oracle. Suppose the success probability of $\mathcal{F}$ is bounded by $\mathsf{Adv}_{\mathsf{FROST2+}}^{\mathsf{ts\text{-}suf\text{-}4}}(\mathcal{F}) \le \Pr[E_1] + \Pr[E_2]+ \Pr[E_3]$. Then, there exist three OMDL adversaries $\mathcal{B}_1',\mathcal{B}_2'$ and $\mathcal{B}_3'$ such that:
\begin{align*}
\mathsf{Adv}_{\mathsf{FROST2+}}^{\mathsf{ts\text{-}suf\text{-}4}}(\mathcal{F}) 
\le \sqrt{8q \cdot \mathsf{Adv}_{\G}^{\mathsf{omdl}}(\mathcal{B}_1')} + 8 \mathsf{Adv}_{\mathbb{G}}^{\mathsf{omdl}}(\mathcal{B}_2') + 8 \mathsf{Adv}_{\mathbb{G}}^{\mathsf{omdl}}(\mathcal{B}_3') + \mathsf{negl}
\end{align*}
where $\mathcal{B}_1'$ and $\mathcal{B}_3'$ makes at most $2q_s+t$ queries to the \textsc{Chall} oracle, and $\mathcal{B}_2'$ makes at most $2q_s$ queries. 
Moreover, the running times of $\mathcal{B}_1',\mathcal{B}_2'$ and $\mathcal{B}_3'$ are bounded above by
\begin{align*}
T(\mathcal{B}_1') &\le \frac{128t^2q^2 \ln(64tq/\Pr[E_1]^2)}{\Pr[E_1]^2}\cdot T(\F{}) +T(q), \\
T(\mathcal{B}_2') &\le \frac{16\mathsf{ns}^2q \ln(16\mathsf{ns}/\Pr[E_2])}{\Pr[E_2]}\cdot T(\F{})+T(q), \text{ and}\\
T(\mathcal{B}_3') &\le \frac{24 q \ln(24 /\Pr[E_3])}{\Pr[E_3]}\cdot T(\F{}) +T(q),
\end{align*}
respectively, where $t\le \mathsf{ns}$, and $T(q)$ denotes the time required for the relevant group operations and exponentiations.
\end{theorem}

\begin{proof}
We show how to construct the $\mathcal{B}_1'$ OMDL adversary before showing the $\mathcal{B}_3'$ OMDL adversary which shares the same proving strategy.\\

\noindent
\textbf{The $\mathcal{B}_1'$ Adversary.}
We define an algorithm $\mathcal{D}$ which can be viewed as the DKG-counterpart of $\mathcal{B}_1$ in Lemma \ref{lem:tssuf4_1}. 
It constructs a wrapper \C and simulates the oracles in a similar way.
However, \C now defines $f_i(Z)$ for every honest signer instead of querying the OMDL oracles to implicitly define $f(Z)$. 
\C also defines an additional $\textsc{RO}[\mathsf{H_{kg}}]$ oracle to handle the hash queries incurred by \textsf{PedPoP}:\\

\noindent
$\textsc{RO}[\mathsf{H_{kg}}](x):$ If $\mathsf{H_{kg}}(x)\neq \perp$, return $\mathsf{H_{kg}}(x)$. 
Else, \C sets $\mathsf{ctr}_{pp}\gets\mathsf{ctr}_{pp}+1$ and $\mathsf{H_{kg}}(x)\gets h_{\mathsf{ctr}_{pp}}$ then returns $\mathsf{H_{kg}}(x)$.\\

Particularly, in $\widetilde{\textsc{Init}}(CS)$, let honest signers be the set $HS=\{\mathsf{id}_i,\ldots,\mathsf{id}_\mathsf{ns}\}\setminus CS$, \C queries 
$\{A_{\ell}=g^{a_{\ell}}\}_{\ell=0}^{t-1} \leftarrow \textsc{Chall}$.
\C selects $\{y_i\}_{\mathsf{id}_i\in HS}\xleftarrow{\$}\Zs{p}$
to implicitly define the random polynomials $f_i(Z)=y_i(\sum_{\ell=0}^{t-1} a_{\ell} Z^{\ell})$ for every honest signer $\mathsf{id}_i\in HS$.
Firstly, for every $\mathsf{id}_i\in HS$, \C increments $\mathsf{ctr}_{pp}$ and chooses $\bar{z}_i\xleftarrow{\$}\Zs{p}$ to compute $\bar{R}_i=g^{\bar{z}_i}/A_0^{y_ih_{\mathsf{ctr}_{pp}}}$ and sets $\mathsf{H_{kg}}(A_{i,0}=A_0^{y_i},A_{i,0},\bar{R}_i)\leftarrow h_{\mathsf{ctr}_{pp}}$.
\C broadcasts $((\bar{R}_i,\bar{z}_i),\mathbf{C}=(A_{i,0}=A_0^{y_i},\ldots,A_{i,t-1}=A_{t-1}^{y_i}))$ as the proof of possession (PoP).
Secondly, \C queries $f'(\mathsf{id}_j)\leftarrow\textsc{Dlog}(\prod_{\ell=0}^{t-1}A_{\ell}^{\mathsf{id}_j^{\ell}})$ and sends $x_{i,j}=y_if'(\mathsf{id}_j)$ on behalf of every honest signer $\mathsf{id}_i\in HS$ to every corrupted signer $\mathsf{id}_j\in CS$.
In return, \C receives $x_{j,i}=f_j(\mathsf{id}_i)$ from $\mathsf{id}_j\in CS$ for every $\mathsf{id}_i\in HS$.
So far, \C issued $t$ \textsc{Chall} queries and $|CS|<t$ \textsc{Dlog} queries.

As \C simulates the DKG perfectly, all PoP broadcast by $HS$ are verified and every corrupted signer $\mathsf{id}_j\in CS$ obtains their signing key share $x_j=f(\mathsf{id}_j)=\sum_{i=1}^{\mathsf{ns}}f_i(\mathsf{id}_j)=\sum_{i=1}^{\mathsf{ns}}x_{i,j}$ as expected, and the joint public key is $\tilde{X}=g^{a_0\sum_{\mathsf{id}_i\in HS} y_i+\sum_{\mathsf{id}_j\in CS}a_{j,0}}=\prod_{\mathsf{id}_i\in HS}A_{i,0}\prod_{\mathsf{id}_j\in CS}A_{j,0}$.
Furthermore, from \F{}'s view, each signer’s state $\mathsf{st}$ is indistinguishable from that generated in the CKG setting, so the oracle simulations by $\mathcal{D}$ after the key generation stage are identical to that by $\mathcal{B}_1$ in Lemma \ref{lem:tssuf4_1}.

Therefore, from $(i_{sig},\mathsf{Out})\gets\mathsf{Fork}^{\C}$, $\mathcal{D}$ can extract the shared signing key $f(0)=a_0\sum_{\mathsf{id}_i\in HS} y_i+\sum_{\mathsf{id}_j\in CS}a_{j,0}$ and singing key shares $\{f(\mathsf{id}_i)\}_{\mathsf{id}_i\in S_2(lr^*)}$.
These are not sufficient for $\mathcal{D}$ to solve its OMDL challenges.
Specifically, the OMDL solutions required here are no longer the $f(\cdot)$ coefficients (the OMDL solutions of $\mathcal{B}_1$ in Lemma \ref{lem:tssuf4_1}), but the $f'(\cdot)$ coefficients.
Moreover, the values $\{a_{j,0}=f_j(0)\}_{\mathsf{id}_j\in CS}$ are not known to $\mathcal{D}$.

\paragraph{Construct helper $\mathcal{D}'$ from $\mathsf{MIMFork}^{\mathcal{D}}_{|CS|,1}$.}
Denote the output of $\mathcal{D}$ as $(I,\{\mathsf{out}_i\}_{i\in I},\theta)$ where 
$\mathsf{out}_i$ is the internal state of $\mathcal{D}$ up to the $i$-th random oracle response such that $I=(i_1,\ldots,i_{|CS|})$ are the set of random oracle indices for $h_i \leftarrow \mathsf{H_{kg}}(\mathsf{id}_i,A_{i,0}, A_{i,0}, \bar{R}_i)$ during $\widetilde{\textsc{Init}}$ and $\theta=f(0)$.
We now construct a helper algorithm $\mathcal{D}'$
which is a syntactically valid OMDL adversary and it runs $\mathsf{MIMFork}^{\mathcal{D}}_{|CS|,1}$.
$\mathcal{D}'$ generates uniformly random $h_1,\ldots,h_{2q}\in\Zs{p}$ and queries the $2q_s+t$ challenges $(\mathbf{A},\mathbf{R},\mathbf{S})\gets\textsc{Chall}$ to initiate $\mathsf{MIMFork}^{\mathcal{D}}_{|CS|,1}$. 
All \textsc{Dlog} queries from $\mathcal{D}$ are relayed to $\mathcal{D}'$.

Therefore, from $(\mathcal{O},\mathcal{O}')\gets\mathsf{MIMFork}^{\mathcal{D}}_{|CS|,1}$, $\mathcal{D}'$ obtains $f(0)$ from the first run of $\mathcal{D}$, and extracts the $\{f_j(0)=a_{j,0}\}_{\mathsf{id}_j\in CS}$ from corrupted signers' PoPs $((\bar{R}_j, \bar{z}_j=\bar{r}_j + h_j a_{j,0})\gets\mathcal{O}, (\bar{R}_j, \bar{z}_j'=\bar{r}_j + h_j' a_{j,0})\gets\mathcal{O}')$.
Now, $\mathcal{D}'$ can compute
$
f'(0)=a_0=\left(f(0)-\sum_{\mathsf{id}_j\in CS} f_j(0)\right)/\sum_{\mathsf{id}_i\in HS} y_i.
$
With $|CS|+1=t$ values of $(f'(0),\{f'(\mathsf{id}_j)\}_{\mathsf{id}_j\in CS})$,
$\mathcal{D}'$ can recover $f'(\cdot)$ and then $\{a_{\ell}\gets\mathsf{DL}_{\G,g}(A_{\ell})\}_{\ell=1}^{t-1}$ through Lagrange interpolation.
If $|CS|<t-1$, the interpolation still works because $|S_2(lr^*)|<t-|CS|$, and $\mathcal{D}'$ can compute $f'(\mathsf{id}_i)=\left(f(\mathsf{id}_i)-\sum_{\mathsf{id}_j\in CS}f_j(\mathsf{id}_i)\right)/\sum_{\mathsf{id}_i\in HS}y_i$ from \C's internal state for $\mathsf{id}_i\in S_2(lr^*)$.
Hence, $\mathcal{D}'$ solved all $t$ challenges $\{a_{\ell}\gets\mathsf{DL}(A_{\ell})\}_{\ell=0}^{t-1}$ used during $\widetilde{\textsc{Init}}$ and it can also solve the $2q_s$ challenges $\{r_{\ell}\gets\mathsf{DL}(R_{\ell}),s_{\ell}\gets\mathsf{DL}(S_{\ell})\}_{\ell=1}^{q_s}$ used in \textsc{PPO} queries as $\mathcal{B}_1$ done in Lemma \ref{lem:tssuf4_1}.
Since the \textsf{BadPPO} and hash collision happen with a negligible probability, the success probability of $\mathcal{D}'$ is:
\begin{align*}
\mathsf{acc}(\mathcal{D}')\ge\mathsf{acc}(\mathsf{MIMFork}^{\mathcal{D}}_{|CS|,1}) -\mathsf{negl} \ge \mathsf{acc}(\mathsf{Fork}^{\C})/8 -\mathsf{negl}
\ge \Pr[E_1]^2/8q  -\mathsf{negl}
\end{align*}

However, rewinding $\mathcal{D}$ for at most $|I|\cdot T_{max}$ times making $\mathcal{D}'$ issued more than $2q_s$ \textsc{Dlog} queries to answer \textsc{PSignO} queries and so $\mathcal{D}'$ does not solve OMDL.

\paragraph{Construct $\mathcal{B}_1'$ from $\mathcal{D}'$.}
We now construct an OMDL adversary $\mathcal{B}_1'$ which is defined the same as $\mathcal{D}'$ except that $\mathcal{B}_1'$ aborts any executions of \C after it obtains all PoPs $(\bar{R}_j, \bar{z}_j)$ which \F{} sent during the $\widetilde{\textsc{Init}}$ query.
This makes sure $\mathcal{D}'$'s internal state recorded the required $(\theta,\{(\bar{R}_j, \bar{z}_j, \bar{z}_j')\}_{\mathsf{id}_j\in CS})$, yet not issuing any additional \textsc{Dlog} query during the rewinds as $\mathcal{D}'$ suffered from.
Therefore, $\mathcal{B}_1'$ has the same success probability as $\mathcal{D}'$ and it runs in time:
\begin{align*}
T(\mathcal{B}_1')\le T(\mathcal{D}') &\le |I|\cdot T_{max}\cdot T(\mathcal{D}) +T(q)\\
&\le t\cdot T_{max}\cdot 2T(\F{}) +T(q)\\
&\le \frac{8t^2q \ln(8t/\mathsf{acc}(\mathsf{Fork}^{\C}))}{\mathsf{acc}(\mathsf{Fork}^{\C})}\cdot 2T(\F{}) +T(q)\\
&\le \frac{128t^2q^2 \ln(64tq/\Pr[E_1]^2)}{\Pr[E_1]^2}\cdot T(\F{}) +T(q)
\end{align*}
such that $\mathsf{Adv}_{\G}^{\mathsf{omdl}}(\mathcal{B}_1')\ge\Pr[E_1]^2/8q  -\mathsf{negl}$.\\

\noindent
\textbf{The $\mathcal{B}_3'$ Adversary.}
As explained previously, $\mathcal{B}_3'$ simulates the distributed key generation polynomial $f(Z)=f'(Z)\sum_{\mathsf{id}_i\in HS} y_i + \sum_{\mathsf{id}_j\in CS} f_j(Z)$ the same as that of $\mathcal{B}_1'$ and proceeds to simulate the oracles as in Lemma \ref{lem:tssuf4_3}. 
From $\{\mathcal{O},\mathcal{O}'\}\gets\mathsf{MIMFork}_{1,3}^{\C}$, (the helper algorithm $\mathcal{C}'$ and subsequently) $\mathcal{B}_3'$ can extract $f(\mathsf{id}_{i^*})$ to compute the value $f'(\mathsf{id}_{i^*})=\left( f(\mathsf{id}_{i^*}) - \sum_{\mathsf{id}_j\in CS} f_j(\mathsf{id}_{i^*})\right)/\sum_{\mathsf{id}_i\in HS} y_i$
where $\{f_j(\mathsf{id}_{i^*})\}_{\mathsf{id}_j\in CS}$ can be obtained from \C's internal state.
Without loss of generality, assume $|CS|=t-1$ and $\mathcal{B}_3'$ knows exactly $|CS|+1 = t$ point-value pairs $(\mathsf{id}_i,f'(\mathsf{id}_i))$.
If $|CS|<t-1$, $\mathcal{B}_3$ samples a random subset $D\subseteq HS\setminus \{\mathsf{id}_{i^*}\}$ of size $t-1-|CS|$ to compute $f'(\mathsf{id}_i)=x_i\gets\textsc{Dlog}(\prod_{\ell=0}^{t-1} A_{\ell}^{\mathsf{id}_i^{\ell}})$ for each $\mathsf{id}_i\in D$.
Using these $t$ point-value pairs, $\mathcal{B}_3'$ can interpolate $f'(\cdot)$ to extract all coefficients $a_0\gets\mathsf{DL}(A_0),\ldots,a_{t-1}\gets\mathsf{DL}(A_{t-1})$.
For the OMDL challenges used as tokens $(R_i,S_i)$ in the game, $\mathcal{B}_3'$ solves them identically to Lemma \ref{lem:tssuf4_3}.
Consequently, the success probability and running time of $\mathcal{B}_3'$ match those established in Lemma \ref{lem:tssuf4_3}.
\qed
\end{proof}

\section{Performance Evaluations}
\label{sec:performance}
We benchmark the proposed \textsf{FROST2}\texttt{+} against ZCash's \textsf{FROST} implementation (GitHub commit \texttt{3ffc19d}) \cite{zfnd2023speedingFROST} on a laptop with an Intel i9-14900HX (2.20 GHz), 32 GB RAM, running Windows 11 Home Single Edition. 
To support our (and \textsf{FROST2}'s) \textsf{CompPar} algorithm, which differs from the original \textsf{FROST} design, we modified the functions \texttt{binding\_factor\_preimages} and \texttt{compute\_group\_commitment} in \texttt{frost-core/src/lib.rs}.

We benchmark the related schemes by measuring only their signing throughput. 
Variants such as \textsf{ROAST} \cite{ROAST} and \textsf{Olaf} \cite{CGRS23} build on \textsf{FROST2} and add robustness mechanisms.
When those mechanisms are not invoked, their signing cost is essentially that of \textsf{FROST2}. 
Similarly, \textsf{ATS[FROST,DS]} is \textsf{FROST} plus an additional (Schnorr) signature per pre-processing token, so it is dominated by the measured \textsf{FROST} cost. 
We therefore benchmark \textsf{FROST}, \textsf{FROST2}, \textsf{ms-FROST2} and \textsf{FROST2}\texttt{+}, which represent the main computational differences in Table \ref{tab:TS_comparison}, using the Ed25519 curve across various predefined $(t=0.7n,n)$ parameter sets. 
For each configuration, we executed 100 rounds and recorded the average time per signature generation.
Each timing experiment was conducted twice: once with pre-computations enabled and once with them disabled.
When pre-computations were enabled, the timings for token verification and the pre-processing round were excluded from the signature generation time; when disabled, they were included.
The full numerical results are provided in Table~\ref{tab:benchmark} (Appendix~\ref{apdx:benchmark_data}), and Figure~\ref{fig:comparisons} visualizes the signatures-per-second for the four schemes.
We also present a comparison of the algorithmic complexity which focuses on scalar multiplications, and security of our schemes against \textsf{FROST} and its two-round variants in Table~\ref{tab:TS_comparison}.
This analysis highlights the core computational differences in scalar operations, with pre-processing contributions denoted in parentheses.

\begin{figure*}[h!]
    \centering
    \includegraphics[width=0.9\linewidth]{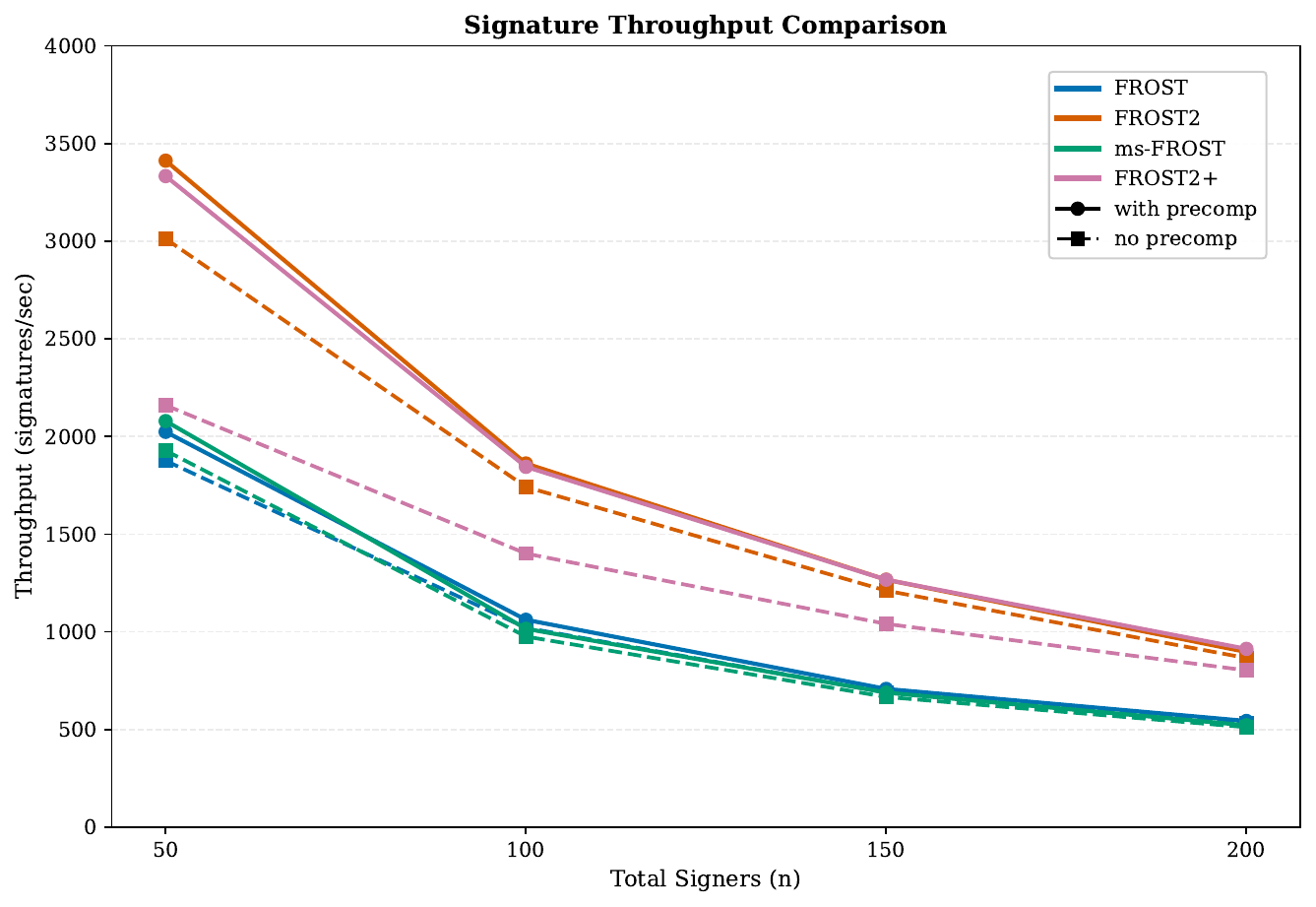}\vspace{-0.5cm}
    \caption{Signatures generated per seconds for \textsf{FROST}, \textsf{FROST2}, \textsf{ms-FROST}, \textsf{FROST2}\texttt{+}.}
    \label{fig:comparisons}
\end{figure*}

\begin{table}[h!]
\centering
\caption{Comparison of complexity andsecurity for the proposed \textsf{FROST}\texttt{+} against \textsf{FROST} and its variants with two signing rounds.}
\label{tab:TS_comparison}
\begin{tabular}{l@{}cclc@{ }c@{}c@{ }cc@{ }c@{ }c@{ }c@{ }}
\toprule
\multirow{2}{*}{Scheme} &  
\multicolumn{5}{c}{Sign} & 
\multirow{2}{*}{\textsf{Agg}} & 
\multirow{2}{*}{\textsf{Vf}} &
\multicolumn{4}{c}{Security}\\
\cmidrule(lr){2-6}
\cmidrule(lr){9-12}
& Round & $|pp|$ & \multicolumn{1}{c}{$smul$} & \G & \Z{p} & $smul$ & $smul$ & Notion & Model & Corrupt & \textsf{Kg} \\
\midrule
\textsf{FROST} \cite{FROST} & 2 & 2 & $(2)+t$ & 2 & 1 & $t$ &2 &\textsf{TS-UF-1} & ROM & Static & D \\
\textsf{FROST} \cite{BCKMTZ22} & 2 & 2 & $(2)+t$ & 3 & 1 & 0 &2 &\textsf{TS-SUF-3} & ROM & Static & C \\
\textsf{FROST}$^{**}$ \cite{CKKTZ25,CKM25} & 2 & 2 & $(2)+t$ & 2 & 1 & $t$ &2 &\textsf{adp-TS-UF} & AGM & Adaptive & C \\
\textsf{ATS[TS, DS]$^*$} \cite{BCKMTZ22} & 2 & 4 & $(3)+3t$ & 4 & 2 & 0 &2 &\textsf{TS-SUF-4} & ROM & Static & C \\
\textsf{FROST2} \cite{CKM21} & 2 & 2 & (2)+1 & 2 & 1 & 1 &2 &\textsf{TS-UF-1} & ROM & Static & D \\
\textsf{FROST2} \cite{BCKMTZ22} & 2 & 2 & (2)+1 & 3 & 1 & 0 &2 &\textsf{TS-SUF-2} & ROM & Static & C \\
\textsf{ROAST} \cite{ROAST} & 2 & 2 & (2)+1 & 2 & 1 & 1 &2 &\textsf{TS-SUF-2} & ROM & Static & C \\
\textsf{Olaf} \cite{CGRS23} & 2 & 2 & (2)+1 & 2 & 1 & 1 &2 &\textsf{TS-SUF-2} & ROM & Static & D \\
\textsf{Arctic} \cite{KG25} & 2 & 1 & $(1)+t^2+t$ & 1 & 1 & $t$ &2 &\textsf{UF} & ROM & Static & C,D \\
\textsf{ms-FROST} \cite{BCLTZ25} & 2 & 2 & (2)+1 & 3 & 1 & 0 &2 &\textsf{adp-TS-UF-4} & AGM & Adaptive & C \\
\textsf{FaFROST} \cite{BGCRSY25} & 2 & 2 & (2)+1 & 3 & 1 & 0 &2 &\textsf{adp-TS-UF} & AGM & Adaptive & C \\
\textsf{FROST-Mask} \cite{Chen25} & 2 & 2 & (2)+1 & 3 & 1 & 1 &2 &\textsf{adp-UF} & ROM & Adaptive & C \\
\textsf{FROST2}$^{**}$ \cite{CKM25} & 2 & 2 & $(2)+1$ & 3 & 1 & 0 &2 &\textsf{adp-TS-UF} & AGM & Adaptive & C \\
\textsf{FROST3}$^{**}$ \cite{CKM25} & 2 & 2 & $(2)+1$ & 2 & 1 & 1 &2 &\textsf{adp-TS-UF} & AGM & Adaptive & C \\
\hline
\textbf{\textsf{FROST2}\texttt{+}} & 2 & 4 & $(3)+4t+1$ & 4 & 2 & 0 &2 &\textsf{TS-SUF-4} & ROM & Static & C,D \\
\textbf{\textsf{FROST2}\texttt{+}}$^{\$}$ & 2 & 4 & $(3+4t)+1$ & 4 & 2 & 0 &2 &\textsf{TS-SUF-4} & ROM & Static & C,D \\
\bottomrule
\multicolumn{12}{l}{$^*$Assume $\mathsf{DS}$ is the Schnorr signature scheme. $^{\$}$with pre-computations.}\\
\multicolumn{12}{l}{$^{**}$ Maximum of $t/2$ adaptive corruptions if in ROM.}\\
\multicolumn{12}{l}{Note: $smul:$ scalar multiplication, $\G,\Z{p}:$ numbers of group and field element sent per signer,}\\
\multicolumn{12}{l}{\textsf{Kg}: key generation, C: centralized, D: decentralized, $t:$ threshold, $(num):$ op. in pre-proc. round}
\end{tabular}
\end{table}

When pre-computations are enabled, \textsf{FROST2}\texttt{+} achieves throughput nearly identical (within -2.3\% to 1.8\%) to \textsf{FROST2} across all group sizes, demonstrating that its additional features come with negligible overhead.
Besides, \textsf{FROST2}\texttt{+} is always $64.7\%=\frac{3333-2024}{2024}\times 100$ (at $n=50$) to $79\%=\frac{1267-708}{708}\times 100$ (at $n=150$) faster than \textsf{FROST} with pre-computations, confirming the better scalability of the underlying algorithms in \textsf{FROST2} and \textsf{FROST2}\texttt{+}.
Specifically, at $n = 200$ signers, \textsf{FROST2}\texttt{+} reaches 913 signatures/sec, marginally surpassing \textsf{FROST2} (897 sig/s) and outperforming \textsf{FROST} (544 sig/s). 
\textsf{ms-FROST} shows a similar performance to \textsf{FROST}, suggesting that the hashing operation is also an important factor for signing speed.

When pre-computations are disabled, \textsf{FROST2}\texttt{+} exhibits a noticeable performance degradation compared to \textsf{FROST2} (e.g., 2160 vs. 3012 sig/s at $n=50$). 
This suggests that the optimizations in \textsf{FROST2}\texttt{+} rely heavily on the pre-computations to achieve competitive throughput, and without them the additional operations from authenticating tokens impose a measurable cost.
Nonetheless, even without pre-computations, \textsf{FROST2}\texttt{+} outperforms \textsf{FROST} for all tested $n$, showing the efficiency of the lighter \textsf{CompPar} algorithm.

\section{Conclusion}
\label{sec:conclude}
We presented an enhanced \textsf{FROST2} scheme, namely, \textsf{FROST2}\texttt{+} with proven \textsf{TS-SUF-4} security under OMDL assumption in the ROM with static corruption and (de)centralized key generation settings.
\textsf{FROST2}\texttt{+} is the only ROM-based \textsf{TS-SUF-4} scheme to date apart from Bellare et al.'s \textsf{ATS[FROST,DS]} where the primitive \textsf{FROST} used there is at least $64\%$ slower than \textsf{FROST2}\texttt{+} for a signing group with threshold $t=0.7n$. 


\bibliographystyle{splncs04}
\bibliography{references}

\appendix

\makeatletter
\renewcommand{\theHsection}{\Alph{section}}
\renewcommand{\theHfigure}{\Alph{section}.\arabic{figure}}
\renewcommand{\theHtable}{\Alph{section}.\arabic{table}}
\makeatother

\section{Attacks on \textsf{FROST2}}
\subsection{Types of Forgeries}
Referring to the security model in Section \ref{sec:security_model}, we recall the conditions under which a (strong) forgery can be achieved in a threshold signature scheme that is also an echo scheme \cite{BTZ22,BCKMTZ22} (e.g., \textsf{FROST2}) in Figure \ref{fig:s_tree}.
\begin{figure}[t]
\centering
\tikzset{
  decision/.style={draw, rounded corners, fill=blue!10, align=center, inner sep=3pt},
  outcome/.style={draw, rounded corners, fill=green!8, align=center, inner sep=3pt},
  xoutcome/.style={draw, rounded corners, fill=red!8, align=center, inner sep=3pt},
  edgearr/.style={-{Latex[length=3mm]}, thick}
}
\begin{tikzpicture}[
  node distance=10mm and 24mm
]
\node[decision] (root) {$\mathsf{SVf}(\tilde{X},lr,\sigma)=\mathsf{true}$};
\node[decision, below left=10mm and 10mm of root] (L) {$|S_2(lr)| < t - |CS|$\\(\textsf{TS-SUF-2})};
\node[decision, below right=10mm and 10mm of root] (R) {$|S_2(lr)| \ge t - |CS|$};
\draw[edgearr] (root.south) -- (L.north);
\draw[edgearr] (root.south) -- (R.north);
\node[decision, below=18mm of $(L)!0.5!(R)$] (shared) {$S_2(lr) \subset S_4(lr)$};
\draw[edgearr] (R.south) -- (shared.north);
\node[decision, below left=14mm and 20mm of shared] (gt) {$S_3(lr) \subset S_2(lr)$};
\node[decision, below=14mm of shared] (eq) {$S_2(lr) = S_3(lr)$};
\node[decision, below right=14mm and 40mm of shared] (lt) {$S_2(lr) \subset S_3(lr)$};
\draw[edgearr] (shared.south) -- (gt.north);
\draw[edgearr] (shared.south) -- (eq.north);
\draw[edgearr] (shared.south) -- (lt.north);
\node[xoutcome, below=12mm of gt] (out1) {$S_3(lr) \subset S_2(lr) \subset S_4(lr)$\\(\textsf{TS-SUF-3})};
\draw[edgearr] (gt.south) -- (out1.north);
\node[outcome, below=12mm of eq] (out2) {$S_2(lr) = S_3(lr) \subset S_4(lr)$\\(\textsf{TS-SUF-4})};
\draw[edgearr] (eq.south) -- (out2.north);
\node[outcome, below left=12mm and -6mm of lt] (out3) {$S_2(lr) \subset S_3(lr) = S_4(lr)$\\(\textsf{TS-SUF-3})};
\node[outcome, below right=12mm and -6mm of lt] (out4) {$S_2(lr) \subset S_3(lr) \subset S_4(lr)$\\(\textsf{TS-SUF-3})};
\draw[edgearr] (lt.south) -- (out3.north);
\draw[edgearr] (lt.south) -- (out4.north);
\node[decision, right=14mm of R] (eqS24) {$S_2(lr) = S_4(lr)$};
\draw[edgearr] (R.east) -- (eqS24.west);
\node[xoutcome, below=12mm of eqS24] (out5) {$S_3(lr) \subset S_4(lr) = S_2(lr)$\\(\textsf{TS-SUF-3})};
\draw[edgearr] (eqS24.south) -- (out5.north);
\end{tikzpicture}
\caption{Decision tree for queries comparisons and forgery outcomes. \fcolorbox{black}{red!8}{ \strut } marks the algorithmically prohibited forgery types (i.e., $\mathsf{st.mapPP}(pp_i)\ne \perp$ enforced in \textsf{PS}) while \fcolorbox{black}{green!8}{ \strut } marks the possible forgery types.}
\label{fig:s_tree}
\end{figure}

\subsection{\textsf{TS-UF-3} Attack}
\label{apdx:ts-uf-3_attack}
In Bellare et al.'s work \cite{BTZ22,BCKMTZ22}, they provide a security proof of \textsf{TS-SUF-2} and also a \textsf{TS-UF-3} attack for \textsf{FROST2}.
This indicates that the best possible security level \textsf{FROST2} can achieve is only \textsf{TS-SUF-2}. 
We recall Bellare et al.'s \textsf{TS-UF-3} attack in the setting below:
\begin{itemize}
    \item Total signers are $\{P_1,P_2,P_3,P_4\}$.
    \item Signers $\{P_3,P_4\}$ corrupted.
    \item A signing request on signer set $lr.\mathcal{S}=\{1,2,3\}$ for a message $lr.\mathsf{msg}=M$.
    \item The pre-processing token $(R_1,S_1),(R_2,S_2)$ are known to the corrupted signers before they generate theirs.
    \item First signer's partial signature $z_1=r_1+ds_1+c\Delta_{1,\mathcal{S}}(0)x_1$ can be obtained.
    \item Expected combined signature is
    $z = \textsf{dlog}(\tilde{R})+cx=\sum_{i\in\mathcal{S}}r_i+ds_i+c\Delta_{i,\mathcal{S}}(0)x_i$
\end{itemize}
The attack procedure is as follows:
\begin{enumerate}
    \item $P_3$ generates its pre-processing token as
    \[
    R_3=R_1^{\frac{\Delta_{1,\mathcal{S}'}(0)}{\Delta_{1,\mathcal{S}}(0)}-1}R_2^{-1},S_3=S_1^{\frac{\Delta_{1,\mathcal{S}'}(0)}{\Delta_{1,\mathcal{S}}(0)}-1}S_2^{-1}.
    \]
    where $\mathcal{S}'=\{1,3,4\}$. The resulting group pre-processing token will be:
    \[
    \tilde{R}=R_1^{\frac{\Delta_{1,\mathcal{S}'}(0)}{\Delta_{1,\mathcal{S}}(0)}}S_1^{d\frac{\Delta_{1,\mathcal{S}'}(0)}{\Delta_{1,\mathcal{S}}(0)}}=\prod_{i\in \mathcal{S}}R_iS_i^{d}
    \]
    where $\textsf{dlog}(\tilde{R})=r_1\frac{\Delta_{1,\mathcal{S}'}(0)}{\Delta_{1,\mathcal{S}}(0)}+ds_1\frac{\Delta_{1,\mathcal{S}'}(0)}{\Delta_{1,\mathcal{S}}(0)}$ and $d=\mathsf{H}_{\mathsf{non}}(\tilde{X},lr)$.
    \item $P_3$ computes $z_1'=r_1\frac{\Delta_{1,\mathcal{S}'}(0)}{\Delta_{1,\mathcal{S}}(0)}+ds_1\frac{\Delta_{1,\mathcal{S}'}(0)}{\Delta_{1,\mathcal{S}}(0)}+c\Delta_{1,\mathcal{S}'}(0) x_1=\frac{\Delta_{1,\mathcal{S}'}(0)}{\Delta_{1,\mathcal{S}}(0)}z_1$
    \item $P_3$ and $P_4$ forge a valid combined signature for $M$ under $lr.\mathcal{S}$ as
    \[
    z' = \textsf{dlog}(\tilde{R}) + cx = z'_1+c\Delta_{3,\mathcal{S}'}(0) x_3 + c\Delta_{4,\mathcal{S}'}(0) x_4.
    \]
\end{enumerate}

This is a valid \textsf{TS-UF-3} attack because $S_2=\{1\}$ and unequal to $S_3=\{1,2\}$, i.e., $S_2(lr) \subset S_3(lr) = S_4(lr)$ in Figure \ref{fig:s_tree}.
We note that this attack does not work on our schemes because the attacker's pre-processing token $(R_3,S_3)$ is a combination of the honest signers' pre-processing token whose discrete logarithms are not known. Consequently, a valid authentication transcript $(u,V)$ for $(R_3,S_3)$ cannot be generated and the honest signers abort during \textsc{PSignO}.

\subsection{\textsf{TS-UF-4} Attack}
\label{apdx:ts-uf-4_attack}
In Bellare et al.'s work \cite{BTZ22,BCKMTZ22}, they demonstrate a \textsf{TS-UF-4} attack for \textsf{FROST} but not \textsf{FROST2}.
Here, we mount the attack on \textsf{FROST2} as follows:
\begin{itemize}
    \item Total signers are $\{P_5,P_{10},P_{11},P_{15},P_{20}\}$.
    \item Signers $\{P_5,P_{10}\}$ corrupted.
    \item A signing request on signer set $lr.\mathcal{S}=\{11,15,20\}$ for a message $lr.\mathsf{msg}=M$.
    \item The pre-processing token $lr.\mathsf{PP}=\{(R_{11},S_{11}),(R_{15}=g^{r_{15}},S_{15}=g^{s_{15}}),(R_{20}=g^{r_{20}},S_{20}=g^{s_{20}})\})$ and $r_{15},s_{15},r_{20},s_{20}\in\Zs{p}$ are known to the corrupted signers before they generate theirs.
    \item First signer's partial signature $z_{11}=r_{11}+d_{11}s_{11}+c\Delta_{11,\mathcal{S}}(0)x_{11}$ can be obtained.
    \item Expected combined signature is
    $z = \textsf{dlog}(\tilde{R})+cx=\sum_{i\in\mathcal{S}}r_i+ds_i+c\Delta_{i,\mathcal{S}}(0)x_i$
\end{itemize}
The attack procedure is as follows:
\begin{enumerate}
    \item Let $\mathcal{S}'=\{5, 10, 11\}$, the Lagrange coefficients are the same for:
    \[
    \Delta_{11,lr.\mathcal{S}} = \frac{15 \cdot 20}{(15 - 11)(20 - 11)} = \frac{25}{3} = \frac{5 \cdot 10}{(5 - 11)(10 - 11)} = \Delta_{11,\mathcal{S}'}.
    \]
    The resulting group pre-processing token will be:
    \[
    \tilde{R}=R_{11}S_{11}^{d}g^{r_{15}+ds_{15}+r_{20}+ds_{20}}=\prod_{i\in \mathcal{S}}R_iS_i^{d}
    \]
    where $\textsf{dlog}(\tilde{R})=r_{11}+ds_{11}+r_{15}+ds_{15}+r_{20}+ds_{20}$ and $d=\mathsf{H}_{\mathsf{non}}(\tilde{X},lr)$.
    \item $P_{5}$ and $P_{10}$ forge a valid combined signature for $M$ under $lr.\mathcal{S}$ as
    \[
    z' = \textsf{dlog}(\tilde{R}) + cx = z_{11}+r_{15}+ds_{15}+r_{20}+ds_{20}+c\Delta_{5,\mathcal{S}'}(0) x_{5} + c\Delta_{10,\mathcal{S}'}(0) x_{10}.
    \]
\end{enumerate}

This is a valid \textsf{TS-UF-4} attack because $S_2=\{11\}=S_3$ but unequal to $S_4=\{11,15,20\}$, i.e., $S_2(lr) = S_3(lr) \subset S_4(lr)$ in Figure \ref{fig:s_tree}.
We note that this attack does not work on our schemes for the similar reason in Section \ref{apdx:ts-uf-3_attack}.
Although the attacker knows the discrete logarithms for the pre-processing tokens $(R_{15},S_{15}),(R_{20},S_{20})$, it does not know the honest signers' signing key shares which are required to generate the authentication transcripts $(u_{15},V_{15}),(u_{20},V_{20})$ for the tokens.
Therefore, the honest signer will abort during partial signing because of invalid $(u_{15},V_{15}),(u_{20},V_{20})$.

\section{\textsf{PedPoP}}
\label{apdx:pedpop}
We recall \textsf{PedPoP} key generation \cite{FROST,BCKMTZ22} in Figure \ref{fig:peppop_algo}.

\begin{figure}[h!]
    \centering
    \begin{tabular}{|p{1.1\linewidth}|}
    \hline
    \begin{minipage}{1\linewidth}
    $$\mathsf{PedPop.Kg}(\mathsf{par},t,\mathsf{ns})$$
    \begin{enumerate}    
    \item Every signer $\mathsf{id}_i$ chooses a polynomial $f_i(Z)\xleftarrow{\$}\mathbb{Z}_p^{t-1}$:
    \[
    f_i(Z) = a_{i,0} + a_{i,1}Z + \cdots + a_{i,t-1}Z^{t-1}
    \]
    and computes $A_{i,k} = g^{a_{i,k}}$ for $k = 0, \ldots, t - 1$. Let $x_i = a_{i,0}$ and $X_{i,0} = A_{i,0}$. Every $\mathsf{id}_i$ selects $\bar{r}_i \xleftarrow{\$} \mathbb{Z}_p$ and sets $\bar{R}_i \leftarrow g^{\bar{r}_i}$ to compute $\bar{z}_i \leftarrow \bar{r}_i + h_i x_i$ where $h_i \leftarrow \mathsf{H_{kg}}(\mathsf{id}_i,X_{i,0}, X_{i,0}, \bar{R}_i)$. They broadcast $((\bar{R}_i, \bar{z}_i), \mathbf{C}_i)$ where the $\mathbf{C}_i = (A_{i,0}, \ldots, A_{i,t-1})$ is a commitment vector.

    \item Given $((\bar{R}_j, \bar{z}_j), \mathbf{C}_j)$ from each other signer $\mathsf{id}_j$, every signer computes:
    \[
    h_j \leftarrow \mathsf{H_{kg}}(\mathsf{id}_j,A_{j,0}, A_{j,0}, \bar{R}_j)
    \]
    and verifies whether:
    \[
    \bar{R}_j A_{j,0}^{h_j} = g^{\bar{z}_j}.
    \]
    If any checks fail, they disqualify the corresponding signer and continue to the next step.

    \item Every $\mathsf{id}_i$ computes secret shares $\bar{x}_{i,j} = f_i(\mathsf{id}_j)$ and privately sends $\bar{x}_{i,j}$ to signer $\mathsf{id}_j$ for $j = 1, \ldots, \mathsf{ns}$.

    \item Each signer $\mathsf{id}_j$ verifies the shares they received from other signers by checking:
    \[
    g^{\bar{x}_{i,j}} = \prod_{k=0}^{t-1} A_{i,k}^{\mathsf{id}_j^k}
    \]
    If the check fails for some index $i$, $\mathsf{id}_j$ broadcasts a complaint against $\mathsf{id}_i$.

    \item For each complaint by signer $\mathsf{id}_j$ against $\mathsf{id}_i$, $\mathsf{id}_i$ broadcasts the share $\bar{x}_{i,j}$. If any revealed share fails to satisfy the verification equation, or if $\mathsf{id}_i$ does not broadcast a share for a complaining signer, then $\mathsf{id}_i$ is disqualified and their shares are set to 0.

    \item The secret share for every $\mathsf{id}_j$ is then:
    \[
    \bar{x}_j = \sum_{i=1}^{\mathsf{ns}} \bar{x}_{i,j}.
    \]

    \item If $X_{i,0} = X_{j,0}$ for any $i \neq j$, then abort. Else, output the joint public key as:
    \[
    \tilde{X} = \prod_{i=1}^{\mathsf{ns}} X_{i,0}
    \]
\end{enumerate}
    \end{minipage}\\
    \hline
    \end{tabular}
    \caption{The \textsf{PedPoP} algorithm \cite{FROST,BCKMTZ22}.}
    \label{fig:peppop_algo}
\end{figure}

\section{Detailed Benchmark Results}
\label{apdx:benchmark_data}
The detailed benchmark data is in Table \ref{tab:benchmark}.

\begin{table*}[h]
\centering
\caption{Performance comparison of threshold signature schemes (excluding DKG setup). $(t,n)$ denotes the threshold and number of signers.}
\label{tab:benchmark}
\small
\begin{tabular}{l c c c c c}
\toprule
\textbf{Scheme} & \textbf{Pre-comp} & \textbf{$(t,n)$} & \textbf{Signature Gen.} (ms) & \textbf{Throughput} (sig/sec) & \textbf{Security} \\
\midrule
\multirow{8}{*}{\textsf{FROST}} & \multicolumn{1}{c}{\multirow{4}{*}{\ding{51}}} & $(35,50)$ & 0.494 & 2,024 & \multirow{8}{*}{\textsf{TS-SUF-3}} \\

& & $(70,100)$ & 0.941 & 1,062 & \\
& & $(105,150)$ & 1.412 & 708 & \\
& & $(140,200)$ & 1.838 & 544 & \\

\cline{2-5}
& \multicolumn{1}{c}{\multirow{4}{*}{\ding{55}}} & $(35,50)$ & 0.533 & 1,876 & \\
& & $(70,100)$ & 0.979 & 1,021 & \\
& & $(105,150)$ & 1.449 & 690 & \\
& & $(140,200)$ & 1.879 & 532 & \\

\midrule \multirow{8}{*}{\textsf{FROST2}} & \multicolumn{1}{c}{\multirow{4}{*}{\ding{51}}} & $(35,50)$ & 0.293 & 3,412 & \multirow{8}{*}{\textsf{TS-SUF-2}} \\

& & $(70,100)$ & 0.537 & 1,862 & \\

& & $(105,150)$ & 0.789 & 1,267 & \\

& & $(140,200)$ & 1.115 & 897 & \\

\cline{2-5}
& \multicolumn{1}{c}{\multirow{4}{*}{\ding{55}}} & $(35,50)$ & 0.332 & 3,012 & \\
& & $(70,100)$ & 0.574 & 1,742 & \\
& & $(105,150)$ & 0.826 & 1,211 & \\
& & $(140,200)$ & 1.155 & 866 & \\

\midrule
\multirow{8}{*}{\textsf{ms-FROST}} & \multirow{4}{*}{\ding{51}} &  (35,50)  & 0.481 & 2079 &  \multirow{8}{*}{\textsf{adp-TS-UF-4}}\\
 &  & (70,100) & 0.984 & 1016 &  \\
 &  & (105,150)& 1.458 & 689  &  \\
 &  & (140,200)& 1.914 & 522  &  \\
\cline{2-5}
 & \multirow{4}{*}{\ding{55}} & (35,50)  & 0.519 & 1928  &  \\
 &  & (70,100) & 1.025 & 976  &  \\
 &  & (105,150)& 1.498 & 667  &  \\
 &  & (140,200)& 1.954 & 512  &  \\

\midrule 
\multirow{8}{*}{\textsf{FROST2}\texttt{+}} & \multicolumn{1}{c}{\multirow{4}{*}{\ding{51}}} & $(35,50)$ & 0.300 & 3,333 & \multirow{8}{*}{\textsf{TS-SUF-4}} \\

& & $(70,100)$ & 0.542 & 1,845 & \\
& & $(105,150)$ & 0.789 & 1,267 & \\
& & $(140,200)$ & 1.095 & 913 & \\

\cline{2-5}
& \multicolumn{1}{c}{\multirow{4}{*}{\ding{55}}} & $(35,50)$ & 0.463 & 2,160 & \\
& & $(70,100)$ & 0.714 & 1,401 & \\
& & $(105,150)$ & 0.961 & 1,041 & \\
& & $(140,200)$ & 1.245 & 803 & \\
\bottomrule 
\end{tabular} 
\end{table*}

\end{document}